\documentclass[conference,letterpaper]{IEEEtran}

\usepackage{./report_default}
\usepackage{soul}

\newcommand{\lp}{\left(}
\newcommand{\rp}{\right)}
\newcommand{\lb}{\left[}
\newcommand{\rb}{\right]}
\newcommand{\lbp}{\left\{}
\newcommand{\rbp}{\right\}}
\newcommand{\lba}{\left\lvert}
\newcommand{\rba}{\right\rvert}

\newcommand{\mcal}{\mathcal}
\newcommand{\mrm}{\mathrm}

\newcommand{\mbb}{\mathbb}

\newcommand{\lce}{\left\lceil}
\newcommand{\rce}{\right\rceil}

\newcommand{\E}{\mathbb{E}}

\renewcommand{\Pr}{\mathbb{P}}

\makeatletter
\newcommand*{\indep}{%
  \mathbin{%
    \mathpalette{\@indep}{}%
  }%
}
\newcommand*{\nindep}{%
  \mathbin{
    \mathpalette{\@indep}{\not}
  }%
}
\newcommand*{\@indep}[2]{%
  \sbox0{$#1\perp\m@th$}
  \sbox2{$#1=$}
  \sbox4{$#1\vcenter{}$}
  \rlap{\copy0}
  \dimen@=\dimexpr\ht2-\ht4-.2pt\relax
  \kern\dimen@
  {#2}%
  \kern\dimen@
  \copy0 
} 
\makeatother

\newtheorem{theorem}{Theorem}
\newtheorem{corollary}{Corollary}
\newtheorem{lemma}{Lemma}

\newtheorem{property}{Property}

\begin{document}

\allowdisplaybreaks

\twocolumn 

\title{Sequential Change Detection for Learning in Piecewise Stationary Bandit Environments}

\author{\IEEEauthorblockN{Yu-Han Huang and Venugopal V. Veeravalli}
\IEEEauthorblockA{ECE and CSL, The Grainger College of Engineering\\ University of Illinois Urbana-Champaign, Urbana, IL, USA \\ \{yuhanhh2,vvv\}@illinois.edu}}
\maketitle
\begin{abstract}
A finite-horizon variant of the quickest change detection problem is investigated, which is motivated by a change detection problem that arises in piecewise stationary bandits. The goal is to minimize the \emph{latency}, which is smallest threshold such that the probability that the detection delay exceeds the threshold is below a desired low level, while controlling the false alarm probability to a desired low level. When the pre- and post-change distributions are unknown, two tests are proposed as candidate solutions. These tests are shown to attain order-optimality in terms of the horizon. Furthermore, the growth in their latencies with respect to the false alarm probability and late detection probability satisfies a property that is desirable in regret analysis for piecewise stationary bandits. Numerical results are provided to validate the theoretical performance results.
\end{abstract}


\section{Introduction}
\label{sec:Intro}

The problem of quickest change detection (QCD) has been widely studied for its various applications in science and engineering. In a QCD problem, an agent sequentially observes a sequence of noisy samples, whose distribution changes at an unknown time due to a disturbance in the environment. The goal is to detect the change as soon as possible while maintaining a constraint on the false alarm probability. See \cite{poor-hadj-qcd-book-2009,tart-niki-bass-2014,vvv_qcd_overview,xie_vvv_qcd_overview} for books and survey articles on the topic.

The QCD problem naturally arises in piecewise stationary (PS) bandit environments, where the distributions of rewards on the arms changes at certain time steps and remains stationary between consecutive changes  \cite{liu2018change,cao2019nearly,besson2022efficient,dahlin2023controlling,wang2021near,zhou2020near,zhou2020nonstationary}. In our prior work \cite{huang2024high}, we formulated a finite-horizon variant of the QCD problem that is tailored to PS bandits, but under the unrealistic assumption that the pre- and post-change distributions are known. In this QCD problem, the goal is to minimize the \emph{latency}, which is smallest threshold such that the probability that the detection delay exceeds the threshold is below a desired low level, while controlling the false alarm probability to a desired low level. A time-varying threshold Cumulative Sum (TVT-CuSum) test is proposed in \cite{huang2024high} as a candidate solution. The TVT-CuSum test computes the CuSum statistic using knowledge of the underlying distributions, and declares a change whenever the CuSum statistic surpasses a threshold that increases logarithmically with time. However, the TVT-CuSum test is not applicable to PS bandits, since the underlying reward distributions are unknown in PS bandits. The QCD problem with unknown distributions was studied in \cite{lai2010sequential}, where the metrics used are the expected time to false alarm and the expected detection delay, differing from our prior work \cite{huang2024high}. To our knowledge, the closest variant of the QCD problem to our work  appears in \cite{maillard2019sequential}, where the performance metrics are similar to those in \cite{huang2024high}, and the underlying distributions are also unknown. However, the work in \cite{maillard2019sequential} does not address order-optimality with respect to the horizon—a critical aspect when employing change detectors in piecewise-stationary bandit settings \cite{huang2025change}.

For the purpose of developing change detectors applicable to PS bandits, we study the QCD problem introduced in \cite{huang2024high}, under the assumption that the distributions before and after the change are not known to the detector. As a first step in this study, we investigate the case where the distribution before change is known, but the distribution after the change is unknown, and we develop change detectors for this case. We then generalize these change detectors to the setting where the distributions before and after change are both unknown. 


The remainder of the paper is organized as follows: The formulation of the QCD problem we study is given in Section \ref{sec:ProbForm}. In Section~\ref{sec:post-unknown}, we propose change detectors for the QCD problem where the agent does not know the distribution after the change. We then generalize these change detectors to the QCD problem where the distributions before and after the change are both unknown in Section~\ref{sec:pre-post-unknown}. 
Numerical results validating the analysis are given in Section \ref{sec:sim}, and the concluding remarks are presented in Section \ref{sec:Sum}.

\section{Problem Formulation}
\label{sec:ProbForm}

Let $\lbp X_{n}:n \in \lbp1,\dots,T\rbp\rbp$ be a finite set of independent random variables indexed by time step $n$. The agent  observes these random variables sequentially over a finite horizon $T$. At a change-point $\nu\in\mbb{N}$ unknown to the agent, the distribution of the stochastic observations changes, i.e.,
\begin{align}
    X_{n}\sim\begin{dcases}
    f_{0},\;n < \nu\\
    f_{1},\;n\geq \nu
    \end{dcases}.\label{eq:sample_distr}
\end{align}
%
To be specific, before the change-point $\nu$, the observation $X_{n}$ follows the pre-change density $f_{0}$ with respect to some dominating measure $\lambda$. After the change-point $\nu$, $X_{n}$ follows the post change density $f_{1}$ with respect to the same dominating measure. We assume that the pre- and post-change densities are $\sigma^{2}$-sub-Gaussian. In Section \ref{sec:post-unknown}, we assume that the agent knows the pre-change density $f_{0}$ completely but only knows that the post-change density $f_{1}$ is $\sigma^{2}$-sub-Gaussian. In Section \ref{sec:pre-post-unknown}, we assume that the agent only knows that both densities are $\sigma^{2}$-sub-Gaussian. When the pre-change distribution is unknown, the agent needs at least a small number of pre-change observations to estimate $f_{0}$. Hence, following  \cite{lai2010sequential}, we assume the existence of a pre-change window of length $m$, during which no changes occur, i.e., $\nu > m$. This assumption ensures that the agent has at least $m$ pre-change observations. In Section \ref{sec:post-unknown}, because the pre-change distribution is assumed to be known, the agent does not need pre-change samples to learn the density $f_{0}$; thus, $m$ is set to be $0$. In addition, let $\mu_{i}$ be the mean of the pre-change density $f_{i}$ for $i \in \lbp 0,1 \rbp$. We define the change gap $\Delta$ to be the absolute difference between the pre- and post-change mean, i.e., $\Delta \coloneqq \lba \mu_{0} - \mu_{1} \rba$. In this work, given our focus on detecting changes in PS bandits, we assume that the pre- and post-change means are different, i.e., $\Delta > 0$, as change detectors in PS bandits only need to detect changes in the mean reward.

Let $\tau$ be the stopping time of a (causal) change detector, an algorithm used by the agent to detect changes in the stochastic observations. In accordance with notation that is commonly used in the QCD literature, we let $\Pr_{\nu}$ and $\E_{\nu}$ denote the probability measure and the expectation when the change-point occurs at $\nu$. Similarly, we use $\Pr_{\infty}$ and $\E_{\infty}$ to denote the probability measure and the expectation when no change occurs. Employing the metric proposed in \cite{huang2024high}, we define the 
\emph{latency} $d$ as 
%
\begin{equation}
\begin{aligned}
    d\coloneqq\inf\{& n:\in\lbp1,\dots,T\rbp:\Pr_{\nu}\lp\tau\geq\nu+n\rp\leq\delta_{\mathrm{D}},\\
    &\forall\,\nu\in\lbp m+1,\dots,T-n\rbp\}\label{eq:latency}.
\end{aligned}
\end{equation}
where $\delta_{\mathrm{D}}\in\lp0,1\rp$.
Note that $m = 0$ in Section \ref{sec:post-unknown}. Our goal is to minimize the latency under the constraint that the probability of false alarm over the horizon $T$, i.e., $\Pr_{\infty}\lp\tau\leq T\rp$, is small. As a result, the QCD problem can be defined as follows: for some (small) values $\delta_{\mathrm{F}}, \delta_{\mathrm{D}} \in \lp 0, 1 \rp$,
\begin{equation}
\begin{split}
    &\underset{\tau}{\textrm{minimize}} \quad d\\
    &\;\textrm{s.t.}\enspace\Pr_{\infty}\lp\tau\leq T\rp\leq\delta_{\mathrm{F}}\\
&\quad\quad\Pr_{\nu}\lp\tau\geq \nu + d \rp \leq \delta_{\mathrm{D}},\; \forall\, \nu \in \lbp m + 1, \dots, T - d \rbp.
\end{split}
\label{eq:QCD}
\end{equation}
Since the information about the horizon is not always available to the agent in bandit problems, we assume that the agent is oblivious to the horizon. The theoretical lower bound in Theorem 3 in \cite{huang2024high}, which is derived under that assumption that the agent knows both the pre- and post-change distributions,  demonstrates that the optimal solution to $d$ in \eqref{eq:QCD} is $\Omega \lp \log T \rp$. In addition, in the regret analyses of PS bandit algorithms in \cite{besson2022efficient, huang2025change}, $\delta_{\mrm{F}}$ and $\delta_{\mrm{D}}$ are set to $T^{-\gamma}$ for some $\gamma > 1$. Therefore to ensure that overall contribution of the latency to the regret is $\mcal{O} \lp \log T \rp$,  we require a good change detector to satisfy the following property:
\begin{property}\label{proper:good_CD}
The latency $d$ of a good change detector should be $\mcal{O} \lp \log T + \log \lp 1/\delta_{\mrm{F}} \rp + \log \lp 1/\delta_{\mrm{D}} \rp \rp$.
\end{property}



\section{Change Detectors with Unknown Post-Change Distribution}
\label{sec:post-unknown}
In this section we study the QCD problem in \eqref{eq:QCD}, under the assumption that the pre-change distribution is known (and $m=0$), while the post-change distribution is unknown except for the fact that it is $\sigma^2$-sub-Gaussian. In order to develop some insights into the design of good detectors under this assumption, we first revisit \eqref{eq:QCD} for the case where the pre- and post- change densities are known and $m = 0$, which was studied in \cite{huang2024high}.

With $f_0$ and $f_1$ known, we can construct 
 the Cumulative Sum (CuSum) statistic as follows:
\begin{equation}\label{eq:CuSum_statistic}
    C_{n} = \max_{1\leq k \leq n} \sum_{i=k}^n \log \frac{f_{1} \lp X_{i} \rp}{f_{0} \lp X_{i} \rp},\;n\in\lbp1,\dots,T\rbp 
\end{equation}
which satisfies the recursion:
\begin{equation}    \label{eq:CuSum_recursion}
 C_{n}   = \max\lbp C_{n-1}, 0 \rbp + \log \frac{f_{1} \lp X_{n} \rp}{f_{0} \lp X_{n} \rp}
\end{equation}
with $C_{0} = 0$. The TVT-CuSum test proposed in \cite{huang2024high} declares a change whenever the CuSum statistic surpasses a time-varying threshold. In particular, the TVT-CuSum test has a stopping time 
\begin{align}
    \tau_r \coloneqq \inf\lbp n\in\mbb{N}:\;C_n \geq\beta_0 \lp n,\delta_{\mathrm{F}},r\rp\rbp,\;r>1\label{eq:TVT_CuSum}
\end{align}
where 
\begin{equation}
    \beta_0 \lp n,\delta_{\mathrm{F}},r\rp\coloneqq\log\lp\zeta\lp r\rp\frac{n^{r}}{\delta_{\mathrm{F}}}\rp\label{eq:threshold}
\end{equation}
with $\zeta\lp r\rp \coloneqq\sum_{i=1}^{\infty}\frac{1}{i^{r}}$. 


The CuSum statistic, however, requires full knowledge of the pre- and post-change densities for computation. Consequently, the TVT-CuSum test cannot be applied when the information about the post-change density is unavailable. One way to tackle the unknown post-change distribution is to replace the CuSum statistic with the Generalized Likelihood Ratio (GLR) statistic (with unknown post-change distribution) \cite{lai1998information}, which can be defined as follows:
\begin{align}\label{eq:GLR_statistics-unknown-post}
    G_{n}\coloneqq\sup_{1\leq k \leq n}  \log \frac{\sup_{\mu\in\mbb{R}}\prod_{i=k}^{n}f_{\mu}(X_i)}{\prod_{i=k}^{n}f_{\mu_{0}}(X_i)},\;n\in\lbp1,\dots,T\rbp 
\end{align}
where $f_{\mu}$ is the Gaussian density with mean $\mu$ and variance $\sigma^{2}$. We emphasize that in \eqref{eq:GLR_statistics-unknown-post} we replace the pre-change density $f_{0}$ with $f_{\mu_{0}}$, which might be a different from $f_0$.  The reason for making this substitution is that the GLR statistic with $f_{\mu_{0}}$ in the denominator can be expressed in a form with only one supremum, making the computation of the statistic easier. See Lemma \ref{lem:GLR-kl-unknown-post} in Appendix \ref{sec:thm1} for  details.  The statistic  in \eqref{eq:GLR_statistics-unknown-post} can be viewed as the generalization of the CuSum statistic \cite{vvv_qcd_overview}, where the post-change distribution is Gaussian with arbitrary mean and variance $\sigma^{2}$. Since the threshold of the TVT-CuSum test $\beta_{0}$ is $\mcal{O} \lp \log \lp n/\delta_{\mrm{F}} \rp \rp$, it is reasonable to speculate that the threshold for the test using the GLR statistic is also $\mcal{O} \lp \log \lp n/\delta_{\mrm{F}} \rp \rp$. Taking cues from \cite{besson2022efficient,kaufmann2021mixture}, we propose the GLR test for unknown post-change distribution as follows:
\begin{align}\label{def:GLR-unknown-post}
    \tau_{\mrm{GLR}}\coloneqq\inf\lbp n\in\mbb{N}:\; G_{n} \geq\beta_{\mrm{GLR}}\lp n,\delta_{\mathrm{F}}\rp\rbp
\end{align}
where the threshold function is defined as
\begin{align}\label{def:GLR-unknown-post-thres}
\begin{aligned}
    \beta_{\mrm{GLR}}\lp n,\delta_{\mathrm{F}}\rp\coloneqq&3\log\lp1+\log\lp n\rp\rp\\
    &+\frac{5}{4}\log\lp\frac{3n^{3/2}}{\delta_{\mathrm{F}}}\rp+\frac{11}{2}
\end{aligned}
\end{align}
for $n \in \lbp 1,\dots, T \rbp$

A popular alternative of the CuSum statistic is the Shiryaev-Roberts (SR) statistic \cite{shiryaev2007optimal, vvv_qcd_overview}, which can be written as
\begin{equation}\label{eq:SR_statistic}
    S_{n} = \sum_{k = 1}^{n} \prod_{i=k}^n \frac{f_{1} \lp X_{i} \rp}{f_{0} \lp X_{i} \rp},\;n\in\lbp1,\dots,T\rbp
\end{equation}
which also satisfies the following recursion:
\begin{equation}    \label{eq:SR_recursion}
 S_{n}   = \lp S_{n-1} + 1 \rp \frac{f_{1} \lp X_{n} \rp}{f_{0} \lp X_{n} \rp}
\end{equation}
with $S_{0} = 0$. In the same manner as we  generalized the CuSum statistic, we can construct the  Generalized Shiryaev-Roberts (GSR) statistic (with unknown post-change distribution) as follows:
\begin{align}\label{eq:GSR_statistics-unknown-post}
    W_{n}\coloneqq\sum_{k = 1}^{n}  \frac{\sup_{\mu\in\mbb{R}}\prod_{i=k}^{n}f_{\mu}(X_i)}{\prod_{i=k}^{n}f_{\mu_{0}}(X_i)},\;n\in\lbp1,\dots,T\rbp.
\end{align}
Then, the corresponding GSR test for unknown post-change distribution can be defined as follows:
\begin{align}\label{def:GSR-unknown-post}
    \tau_{\mrm{GSR}}\coloneqq\inf\lbp n\in\mbb{N}:\; \log W_{n} \geq\beta_{\mrm{GSR}}\lp n,\delta_{\mathrm{F}}\rp\rbp
\end{align}
where the threshold function is defined as: 
\begin{align}\label{def:GSR-unknown-post-thres}
\begin{aligned}
    \beta_{\mrm{GSR}}\lp n,\delta_{\mathrm{F}}\rp\coloneqq&\beta_{\mrm{GLR}}\lp n,\delta_{\mathrm{F}}\rp + \log n.
\end{aligned}
\end{align}
for $n \in \lbp 1,\dots, T \rbp$.

In the following theorem, we show that both the GLR and GSR tests can effectively control the false alarm probability with a small latency. The proof of the theorem is given in Appendix \ref{sec:thm1}. 

\begin{theorem}[GLR and GSR tests with unknown post-change distribution] \label{thm:unknown-post}
Consider the GLR test in \eqref{def:GLR-unknown-post} and the GSR test in \eqref{def:GSR-unknown-post}. For any $T \in \mbb{N}$ and $\nu \in \lbp 1, \dots, T - d \rbp$, with $d$ defined as
\begin{align} \label{eq:latency-unknown-post}
    d = \lce\frac{2\sigma^{2}}{\Delta^{2}}\lp\sqrt{\beta\lp T,\delta_{\mathrm{F}}\rp}+\sqrt{\log\lp\frac{2}{\delta_{\mathrm{D}}}\rp}\rp^{2}\rce
\end{align}
the stopping times of the tests satisfy $\Pr_{\infty} \lp \tau \leq T \rp \leq \delta_{\mrm{F}}$ and $\Pr_{\nu} \lp \tau \geq \nu + d \rp \leq \delta_{\mrm{D}}$,  
where $\beta = \beta_{\mrm{GLR}}$ for the GLR test and $\beta = \beta_{\mrm{GSR}}$ for the GSR test.
\end{theorem}

Since $\beta_{\mrm{GLR}}\lp T,\delta_{\mathrm{F}}\rp$ and $\beta_{\mrm{GSR}}\lp T,\delta_{\mathrm{F}}\rp$ are $\mcal{O} \lp \log T \rp$, the latency in \eqref{eq:latency-unknown-post} is $\mcal{O} \lp \log T \rp$, matching the theoretical lower bound in Theorem 3 in \cite{huang2024high}. In addition, the latency in \eqref{eq:latency-unknown-post} is $\mcal{O} \lp \log \lp 1/\delta_{\mrm{F}} \rp + \log \lp 1/\delta_{\mrm{D}} \rp \rp$. Hence, the GLR and GSR tests satisfy Property \ref{proper:good_CD}, indicating that these tests are good candidates  for generalization to the QCD problem of interest in PS bandits,  where both the pre- and post-change distributions are unknown.

\section{Change Detectors with Unknown Pre- and Post-Change Distribution}
\label{sec:pre-post-unknown}


In this section we study the QCD problem in \eqref{eq:QCD}, under the setting of interest for PS bandits, in which  both the pre- and post-change distributions are unknown, except for the fact that they are $\sigma^2$-sub-Gaussian.

Recall that $f_{\mu}$ is the density of a Gaussian random variable with mean $\mu$ and variance $\sigma^{2}$. Similar to the GLR statistic in \eqref{eq:GLR_statistics-unknown-post}, we can generalize the CuSum statistic to the GLR statistic (with unknown pre- and post-change distributions) as follows. For $n \in \lbp 1, \dots, T \rbp$,
\begin{equation}
\begin{aligned}
    &\tilde{G}_{n} \coloneqq \\
    &\sup_{1 \leq k \leq n} \log  \frac{\sup_{\mu_{0}' \in \mbb{R}} \sup_{\mu_{1}' \in \mbb{R}} \prod_{i = 1}^{k} f_{\mu_{0}'}\lp X_{i} \rp \prod_{i = k + 1}^{n} f_{\mu_{1}'}\lp X_{i} \rp}{\sup_{\mu\in\mbb{R}} \prod_{i = 1}^{n} f_{\mu} \lp X_{i} \rp}\label{eq:GLR_statistics_pre-post-unknown}.
\end{aligned}
\end{equation}
Then, taking the same cues from \cite{besson2022efficient,kaufmann2021mixture}, we propose the GLR test for unknown pre- and post-change distributions as follows:
\begin{align}\label{def:GLR-unknown-pre-post}
    \tilde{\tau}_{\mrm{GLR}} \coloneqq \inf \lbp n \in \mbb{N}:\; \tilde{G}_{n} \geq \tilde{\beta}_{\mrm{GLR}} \lp n, \delta_{\mathrm{F}} \rp \rbp
\end{align}
where the threshold function is defined as:  
\begin{align}\label{def:GLR-unknown-pre-post-thres}
\begin{aligned}
    \tilde{\beta}_{\mrm{GLR}} \lp n, \delta_{\mathrm{F}} \rp \coloneqq& 6\log\lp1+\log\lp n\rp\rp\\
    &+\frac{5}{2}\log\lp\frac{4n^{3/2}}{\delta_{\mathrm{F}}}\rp + 11.
\end{aligned}
\end{align}
for $n \in \lbp 1,\dots, T \rbp$.

Similarly, we can generalize the SR statistic to the GSR statistic (with unknown pre- and post-change distributions) as follows. For $n \in \lbp 1, \dots, T \rbp$,
\begin{equation}
\begin{aligned}
&\tilde{W}_{n} \coloneqq \\
&\sum_{k = 1}^{n} \frac{\sup_{\mu_{0}' \in \mbb{R}} \sup_{\mu_{1}' \in \mbb{R}} \prod_{i = 1}^{k} f_{\mu_{0}'} \lp X_{i} \rp \prod_{i = k + 1}^{n} f_{\mu_{1}'}\lp X_{i}\rp}{\sup_{\mu \in \mbb{R}} \prod_{i = 1}^{n} f_{\mu} \lp X_{i} \rp}.
\label{eq:GSR_statistics_pre-post-unknown}
\end{aligned}
\end{equation}
We can then construct GSR test for unknown pre- and post-change distributions as follows:
\begin{align}\label{def:GSR-unknown-pre-post}
    \tilde{\tau}_{\mrm{GSR}} \coloneqq \inf \lbp n \in \mbb{N}:\; \log \tilde{W}_{n} \geq \tilde{\beta}_{\mrm{GSR}} \lp n, \delta_{\mathrm{F}} \rp \rbp
\end{align}
where the threshold function is defines as: 
\begin{align}\label{def:GSR-unknown-pre-post-thres}
\begin{aligned}
    \tilde{\beta}_{\mrm{GSR}} \lp n, \delta_{\mathrm{F}} \rp \coloneqq& \tilde{\beta}_{\mrm{GLR}} \lp n, \delta_{\mathrm{F}} \rp + \log n.
\end{aligned}
\end{align}
for $n \in \lbp 1,\dots, T \rbp$.

In the following theorem, we illustrate that the GLR and GSR test with unknown pre- and post-change distributions can also detect changes with a low latency, while ensuring that the false alarm probability is small, given that there are enough pre-change observations. The proof of the theorem is given in Appendix \ref{sec:thm2}. 

\begin{theorem} [GLR and GSR tests with unknown pre- and post-change distributions]
\label{thm:pre-post-unknown}
Consider the GLR test in \eqref{def:GLR-unknown-pre-post} and the GSR test in \eqref{def:GSR-unknown-pre-post}, with 
\begin{equation} \label{eq:m}
    m \geq \frac{8\sigma^{2}}{\Delta^{2}}\beta \lp T,\delta_{\mrm{F}}\rp.
\end{equation}
Then, for $\nu\in\lbp m+1,\dots,T-d\rbp$, with $d$ defined as:
\begin{equation}
  d\coloneqq\lce\max\lbp\frac{8\sigma^{2}m\beta \lp T,\delta_{\mrm{F}}\rp}{\Delta^{2}m-8\sigma^{2}\beta \lp T,\delta_{\mrm{F}}\rp},\frac{\delta_{\mrm{F}}^{2/3}}{2^{16/15}\delta_{\mrm{D}}^{4/15}}-m\rbp\rce
    \label{eq:d}
\end{equation}
the false alarm probability and the latency satisfy $\Pr_{\infty}\lp \tau \leq T \rp \leq \delta_{\mrm{F}}$ and $\Pr_{\nu}\lp\tau\geq\nu+d\rp\leq\delta_{\mrm{D}}$, respectively,
where $\beta \lp T,\delta_{\mrm{F}}\rp = \tilde{\beta}_{\mrm{GLR}} \lp T,\delta_{\mrm{F}}\rp$ in \eqref{def:GLR-unknown-pre-post-thres}
for the GLR test and $\beta \lp T,\delta_{\mrm{F}}\rp = \tilde{\beta}_{\mrm{GSR}} \lp T,\delta_{\mrm{F}}\rp$ in \eqref{def:GSR-unknown-pre-post-thres} for the GSR test.
\end{theorem}


Recall that $\tilde{\beta}_{\mrm{GLR}} \lp T,\delta_{\mrm{F}}\rp$ and $\tilde{\beta}_{\mrm{GSR}} \lp T,\delta_{\mrm{F}}\rp$ are $\mcal{O} \lp \log \lp T/\delta_{\mrm{F}} \rp \rp$. In the following corollary, we show that the GLR and the GSR tests satisfy Property \ref{proper:good_CD} with an appropriate choice of $m$ if $\delta_{\mrm{F}} \leq \delta_{\mrm{D}}$.

\begin{corollary} \label{cor:m.d}
If $\delta_{\mrm{F}} \leq \delta_{\mrm{D}}$, then the GLR and the GSR tests satisfy Property \ref{proper:good_CD} with $m = \lce \frac{16\sigma^{2}}{\Delta^{2}}\beta \lp T,\delta_{\mrm{F}}\rp + \log \lp 1/\delta_{\mrm{D}} \rp \rce$, where $\beta \lp T,\delta_{\mrm{F}}\rp = \tilde{\beta}_{\mrm{GLR}} \lp T,\delta_{\mrm{F}}\rp$ in \eqref{def:GLR-unknown-pre-post-thres}
for the GLR test and $\beta \lp T,\delta_{\mrm{F}}\rp = \tilde{\beta}_{\mrm{GSR}} \lp T,\delta_{\mrm{F}}\rp$ in \eqref{def:GSR-unknown-pre-post-thres} for the GSR test.
\end{corollary}

\begin{proof}
With the choice of $m$, when $\delta_{\mrm{F}} \leq \delta_{\mrm{D}}$, we have
\begin{align}
d &= \lce\max\lbp\frac{8\sigma^{2} m\beta\lp T,\delta_{\mrm{F}}\rp}{\Delta^{2}m-8\sigma^{2}\beta\lp T,\delta_{\mrm{F}}\rp}, \frac{\delta_{\mrm{F}}^{2/3}}{2^{16/15}\delta_{\mrm{D}}^{4/15}} - m\rbp\rce \nonumber\\
&\leq \lce\max\lbp\frac{8\sigma^{2} m\beta\lp T,\delta_{\mrm{F}}\rp}{\Delta^{2} \lce \frac{16\sigma^{2}}{\Delta^{2}}\beta \lp T,\delta_{\mrm{F}}\rp \rce - 8\sigma^{2}\beta\lp T,\delta_{\mrm{F}}\rp}, 1 - m\rbp\rce \nonumber\\
&\leq \lce \max\lbp m, 1 - m \rbp \rce \nonumber \\
&\overset{(a)}{=} \lce \frac{16\sigma^{2}}{\Delta^{2}}\beta \lp T,\delta_{\mrm{F}}\rp + \log \lp 1/\delta_{\mrm{D}} \rp \rce \nonumber \\
&\overset{(b)}{=} \mcal{O} \lp \log T + \log \lp 1/\delta_{\mrm{F}} \rp + \log \lp 1/\delta_{\mrm{D}} \rp \rp
\label{eq:d_upp_glr}
\end{align}
where step $(a)$ follows from the choice of $m$, and (b) is due to the fact that $\beta \lp T, \delta_{\mrm{F}} \rp$ is $\mcal{O} \lp \log T + \log \lp 1/\delta_{\mrm{F}} \rp \rp$.
\end{proof}

As a result, the GLR test in \eqref{def:GLR-unknown-pre-post} and the GSR test in \eqref{def:GSR-unknown-pre-post} both satisfy Property \ref{proper:good_CD}, meaning that they are good change detectors for PS bandits.

\section{Experimental Results}
\label{sec:sim}

In this section, we study the performance of our proposed change detectors through simulations. We compare the latencies of the proposed GLR change detectors with that of the TVT-CuSum test \cite{huang2024high} to illustrate the cost of detecting changes without knowing the (pre- and) post-change distribution(s). To investigate the tightness of the upper bounds in our theorems, we also compare the empirical results for our change detectors with the upper bounds in Theorem \ref{thm:unknown-post} and \ref{thm:pre-post-unknown} under various choices of $T$, $\delta_{\mrm{F}}$, and $\delta_{\mrm{D}}$.

In our experiments, the pre- and post-change distributions follow $\mcal{N}\lp 0, 1 \rp$ and $\mcal{N}\lp 1, 1 \rp$, respectively. When the pre- and post-change distributions are both unknown, we set the pre-change window $m = T - 1000$, since this choice of $m$ is greater than $\lce \frac{16\sigma^{2}}{\Delta^{2}}\beta \lp T,\delta_{\mrm{F}}\rp + \log \lp 1/\delta_{\mrm{D}} \rp \rce$ in Corollary \ref{cor:m.d}, ensuring that the latency of the GLR test in \eqref{def:GLR-unknown-pre-post} satisfies Property \ref{proper:good_CD}. According to the definition of the latency in \eqref{eq:latency}, for any arbitrary change-point $\nu \in \lbp m+1, \dots, T - d \rbp$, there are approximately $100\delta_{\mrm{D}}\%$ of the simulated trials in which the detection delay exceeds the latency. Therefore, to obtain the empirical value of the latency, we first conducted $200000$ trials and recorded the detection delay in each trial for each change-point in a set $\mcal{N}$, and then took the maximum of the $100 \lp 1 - \delta_{\mrm{D}} \rp^{\mrm{th}}$ percentile of the recorded detection delays over all change-points in $\mcal{N}$. The collection of change-points $\mcal{N}$ is set to $\lbp m+1+nT/10 :\; n \in \mbb{N},\; m+1+nT/10 \leq T \rbp$, as conducting $200000$ trials over all possible change-points is cumbersome. 

To implement the GLR tests in \eqref{def:GLR-unknown-post} and \eqref{def:GLR-unknown-pre-post}, we need to compute the GLR statistics in \eqref{eq:GLR_statistics-unknown-post} and \eqref{eq:GLR_statistics_pre-post-unknown}. However, unlike the CuSum statistic in \eqref{eq:CuSum_statistic}, the GLR statistics in \eqref{eq:GLR_statistics-unknown-post} and \eqref{eq:GLR_statistics_pre-post-unknown} do not have the recursive property to simplify the computation. As a result, we perform down-sampling by taking supremum over $\mcal{K}_{n} \coloneqq \lbp n - 700, \dots, n \rbp$ when computing the GLR statistics in \eqref{eq:GLR_statistics-unknown-post} and \eqref{eq:GLR_statistics_pre-post-unknown}, i.e., for $n \in \lbp 1, \dots, T \rbp$,
\begin{align}
    &G'_{n}\coloneqq \sup_{k \in \mcal{K}_{n}}  \log \frac{\sup_{\mu\in\mbb{R}}\prod_{i=k}^{n}f_{\mu}(X_i)}{\prod_{i=k}^{n}f_{\mu_{0}}(X_i)}, \\
    &\tilde{G}'_{n} \coloneqq \nonumber \\
    &\sup_{k \in \mcal{K}_{n}} \log \frac{\sup_{\mu_{0}' \in \mbb{R}} \sup_{\mu_{1}' \in \mbb{R}} \prod_{i = 1}^{k} f_{\mu_{0}'}\lp X_{i} \rp \prod_{i = k + 1}^{n} f_{\mu_{1}'}\lp X_{i} \rp}{\sup_{\mu\in\mbb{R}} \prod_{i = 1}^{n} f_{\mu} \lp X_{i} \rp} .\label{eq:GLR_stat_pre-post-unknown-down}
\end{align}
Then, the stopping time of the implemented GLR tests can be defined as follows:
\begin{align}
    \tau'_{\mrm{GLR}}&\coloneqq\inf\lbp n\in\mbb{N}:\; G'_{n} \geq\beta_{\mrm{GLR}}\lp n,\delta_{\mathrm{F}}\rp\rbp \label{def:GLR_unknown-post-down},
    \\
    \tilde{\tau}'_{\mrm{GLR}}&\coloneqq\inf\lbp n\in\mbb{N}:\; \tilde{G}'_{n} \geq\tilde{\beta}_{\mrm{GLR}}\lp n,\delta_{\mathrm{F}}\rp\rbp \label{def:GLR_pre-post-unknown-down}
\end{align}
where $\beta_{\mrm{GLR}}$ is defined in \eqref{def:GLR-unknown-post-thres} and $\tilde{\beta}_{\mrm{GLR}}$ is defined in \eqref{def:GLR-unknown-pre-post-thres}. However, for the GSR statistics in \eqref{def:GSR-unknown-post} and \eqref{def:GSR-unknown-pre-post}, such down-sampling cannot be implemented. Consequently, we only perform the simulations for the GLR tests in \eqref{def:GLR_unknown-post-down} and \eqref{def:GLR_pre-post-unknown-down}. The empirical results for these tests are given in Figures~\ref{fig:latency_T} and \ref{fig:latency_delta}.
\begin{figure}
    \centering
    \includegraphics[width=0.9\linewidth]{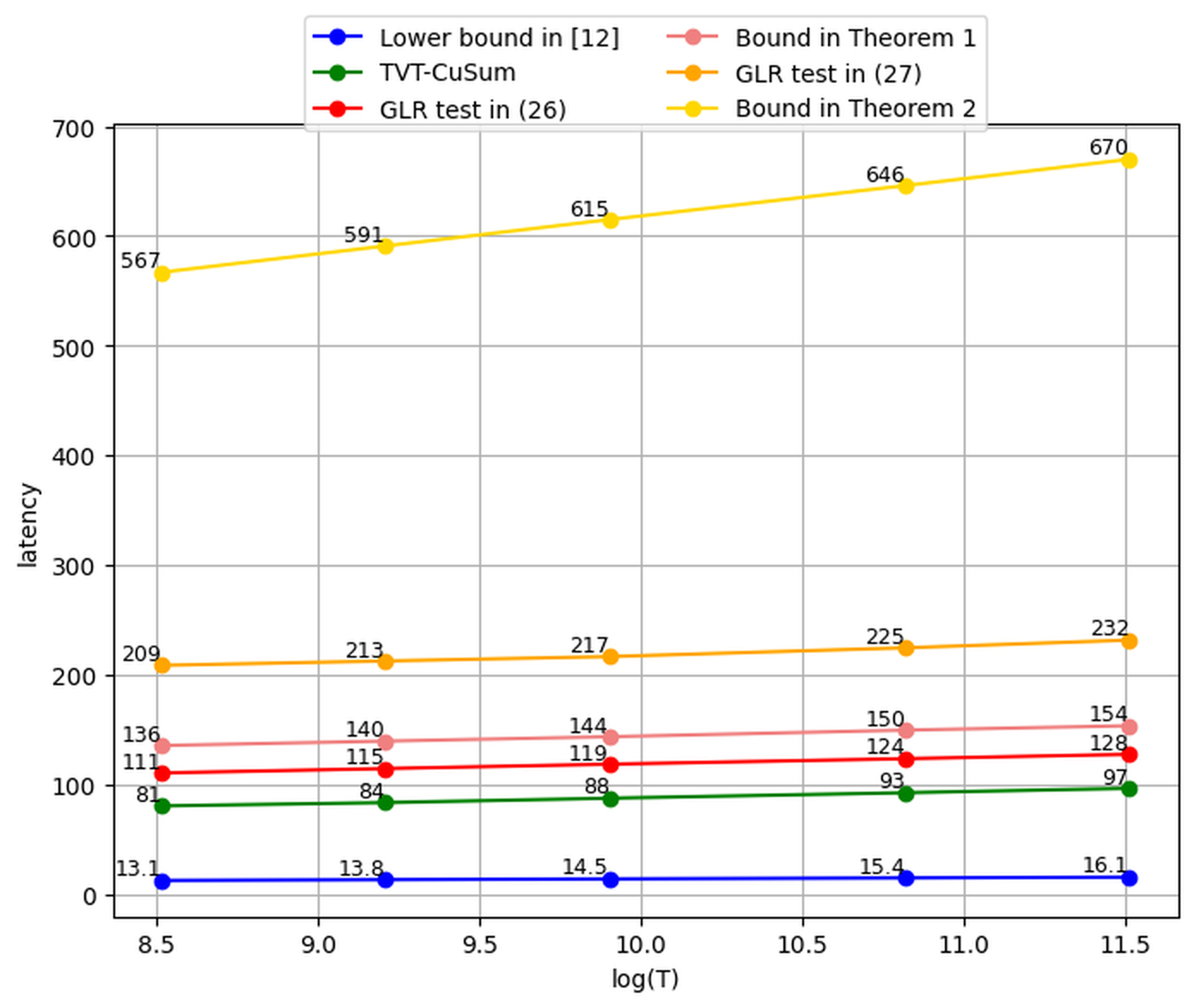}
    \caption{Latencies of $\tau'_{\mrm{GLR}}$ and $\tilde{\tau}'_{\mrm{GLR}}$ with fixed $\delta_{\mrm{F}} = \delta_{\mrm{D}} = 0.01$ and varying $T \in \lbp 5000, 10000, 20000, 50000, 100000 \rbp$.}
    \label{fig:latency_T}
\end{figure}
\begin{figure}
    \centering
    \includegraphics[width=0.9\linewidth]{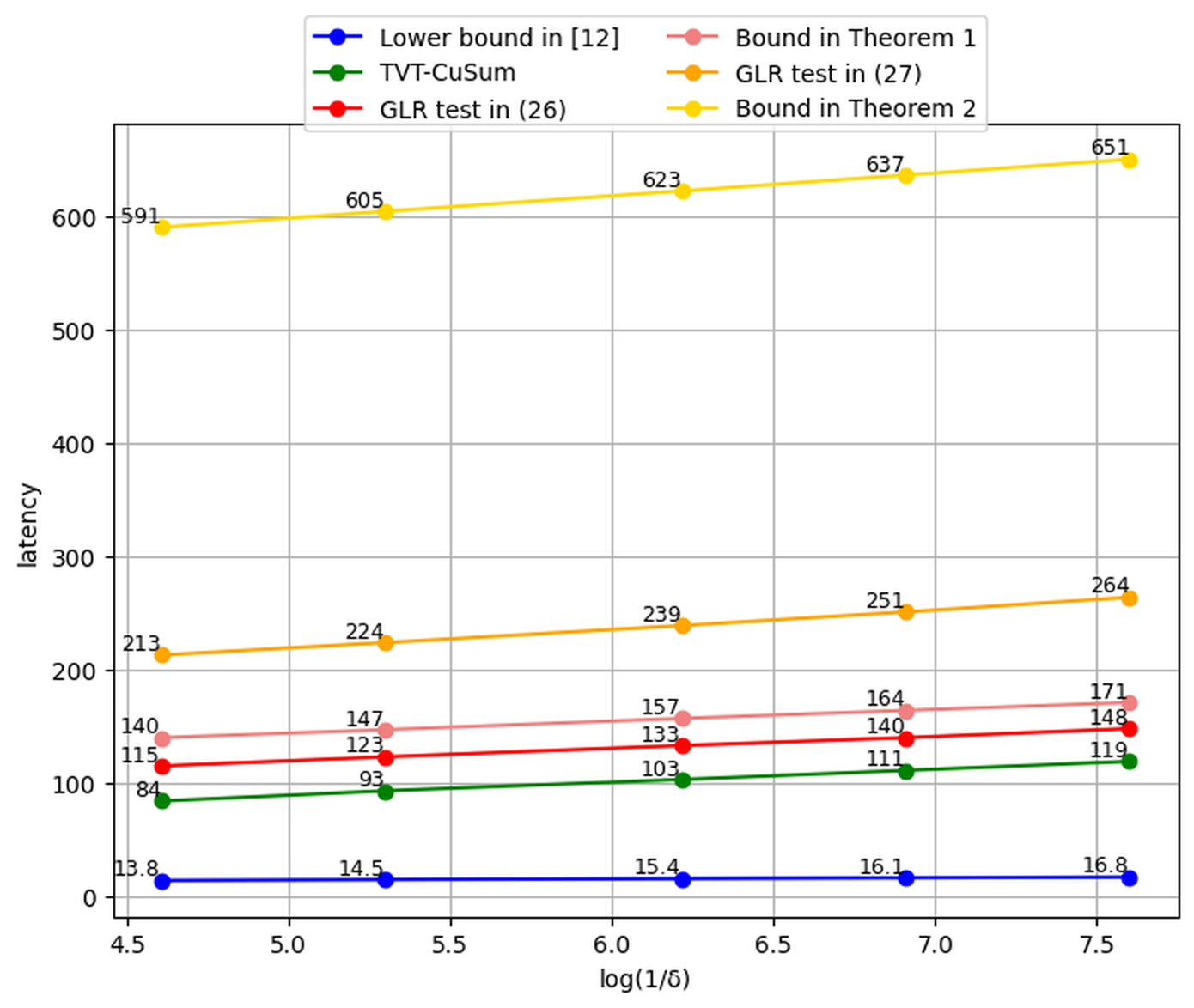}
    \caption{Latencies of $\tau'_{\mrm{GLR}}$ and $\tilde{\tau}'_{\mrm{GLR}}$ with fixed $T = 10000$ and varying $\delta \in \lbp 5000, 10000, 20000, 50000, 100000 \rbp$ where $\delta_{\mrm{F}} = \delta_{\mrm{D}} = \delta$.}
    \label{fig:latency_delta}
\end{figure}

As shown in the figures, the latencies of the GLR test in \eqref{def:GLR_unknown-post-down} and \eqref{def:GLR_pre-post-unknown-down} grow logarithmically with $T$ and $1/\delta$, aligning with the growth of the upper bounds in Theorem \ref{thm:unknown-post} and \ref{thm:pre-post-unknown}. Nevertheless, the upper bound in Theorem \ref{thm:pre-post-unknown} is loose compared to that in Theorem \ref{thm:unknown-post}, as the gap between the simulated value and the upper bound is wider when the pre- and post-change distributions are unknown. 
We can also see that the gap between the latency of the TVT-CuSum test and that of the GLR test in \eqref{def:GLR_unknown-post-down} is narrower than the gap between the latency of the GLR test in \eqref{def:GLR_unknown-post-down} and that of the GLR test in \eqref{def:GLR_pre-post-unknown-down}. 


\section{Conclusions}
\label{sec:Sum}

We investigated a variant of the QCD problem that is designed for the PS bandits. Since the reward distributions are unknown in bandit problems, we considered scenarios where the information about the post-change distribution is unavailable first, and then generalized the results to the cases where both pre- and post-change distributions are unknown. We proposed the GLR and GSR tests for each setting as candidate solutions. We show that the growth of the latency of the GLR and GSR tests with $T$ is order-optimal in both cases. In addition, the latency of the GLR and GSR tests grows logarithmically with $1/\delta_{\mrm{F}}$ and $1/\delta_{\mrm{D}}$, which is a desirable property in the regret analysis for PS bandits.
Our simulation results validate these theoretical results. However, the simulation results suggest that the upper bound on the latency in Theorem \ref{thm:pre-post-unknown} is loose, and it may be worthwhile exploring if this bound can be tightened. 


\section*{Acknowledgement}
This work was supported by the National Science Foundation under grant ECCS-2033900, and by the Army Research Laboratory under Cooperative Agreement W911NF-17-2-0196, through the University of Illinois at Urbana-Champaign.
\vfill

{\footnotesize\bibliography{ref.bib}}

\onecolumn

\appendices

\section{Proof of Theorem \ref{thm:unknown-post}} \label{sec:thm1}

There are two change detectors to consider: the GLR and GSR tests. For each test, there are two parts to prove: the false alarm probability $\Pr_{\infty} \lp \tau \leq T \rp$ and the late detection probability $\Pr_{\nu} \lp \tau \geq \nu + d \rp$. To prove these probability upper bounds, we first express the statistic using the empirical mean of the observations, so that we can exploit the sub-Gaussianity of the samples and apply concentration inequalities in manner similar to the approach in \cite{besson2022efficient, kaufmann2021mixture}.

We first prove the result for the GLR test in \eqref{def:GLR-unknown-post}: Let $\hat{\mu}_{m:n}$ be the empirical mean of the stochastic observations $\lbp X_{m}, \dots, X_{n} \rbp$. Recall that $f_{\mu}$ is the Gaussian density with mean $\mu$ and variance $\sigma^{2}$. To represent the GLR statistic in \eqref{eq:GLR_statistics-unknown-post} in terms of the empirical mean of the observations, we use the following lemma:

\begin{lemma}\label{lem:GLR-kl-unknown-post}
For any $\mu \in \mbb{R}$ and $k, n \in \lbp 1, \dots, T \rbp$ where $k \leq n$,
\begin{align}
    \log \lp \frac{\sup_{\mu' \in \mbb{R}}\prod_{i = k}^{n} f_{\mu'} \lp X_{i} \rp}{\prod_{i = k}^{n} f_{\mu} \lp X_{i} \rp} \rp = \lp n - k + 1 \rp \mathrm{kl}\lp\hat{\mu}_{k:n};\mu\rp.\label{eq:GLR-kl}
\end{align}
where $\mathrm{kl} \lp x; y \rp \coloneqq \frac{\lp x - y \rp^{2}}{2 \sigma^{2}}$ is the KL divergence between two Gaussian distributions with common variance $\sigma^{2}$ and different mean $x$ and $y$, respectively.
\end{lemma}

\begin{proof}

For any $\mu \in \mbb{R}$ and $k, n \in \lbp 1, \dots, T \rbp$ where $k \leq n$, we can show that
\begin{align}
\begin{aligned}
    &\log \lp \frac{\sup_{\mu' \in \mbb{R}}\prod_{i = k}^{n} f_{\mu'} \lp X_{i} \rp}{\prod_{i = k}^{n} f_{\mu} \lp X_{i} \rp} \rp \\
    &=\log \lp \frac{\frac{1}{\lp 2 \pi \sigma^{2} \rp^{n/2}} \exp \lp - \frac{1}{2\sigma^{2}} \inf_{\mu' \in \mbb{R}} \sum_{i = k}^{n} \lp X_{i} - \mu' \rp^{2} \rp}{\frac{1}{\lp 2 \pi \sigma^{2} \rp^{n/2}} \exp \lp - \frac{1}{2\sigma^{2}} \sum_{i = k}^{n} \lp X_{i} - \mu' \rp^{2} \rp} \rp \\
    &\overset{(a)}{=} \log \lp \frac{\frac{1}{\lp 2 \pi \sigma^{2} \rp^{n/2}} \exp \lp - \frac{1}{2 \sigma^{2}} \sum_{i = k}^{n} \lp X_{i} - \hat{\mu}_{k : n} \rp^{2} \rp}{\frac{1}{\lp 2 \pi \sigma^{2} \rp^{n/2}} \exp \lp - \frac{1}{2\sigma^{2}} \sum_{i = k}^{n} \lp X_{i} - \mu \rp^{2} \rp} \rp \\
    &= \frac{1}{2\sigma^{2}} \sum_{i = k}^{n} \lb \lp X_{i} - \mu \rp^{2} - \lp X_{i} - \hat{\mu}_{k:n} \rp^{2} \rb \\
    &= \frac{1}{2\sigma^{2}} \lb \lp n - k + 1 \rp \mu^{2} - 2 \mu \sum_{i = k}^{n} X_{i} - \lp n - k + 1 \rp \hat{\mu}^{2}_{k:n} + 2 \hat{\mu}_{k:n} \sum_{i = k}^{n} X_{i}  \rb \\
    &= \frac{n - k + 1}{2\sigma^{2}} \lp \mu - \hat{\mu}_{k:n} \rp^{2} \\
    &= \lp n - k + 1 \rp \mrm{kl} \lp \hat{\mu}_{k:n}; \mu \rp
\end{aligned}
\end{align}
where step $(a)$ is owing to the fact that $\sum_{i = k}^{n} \lp X_{i} - \mu' \rp^{2}$ is minimized when $\mu' = \hat{\mu}_{k:n}$.
\end{proof}

For upper bounding the false alarm probability of the GLR test in \eqref{def:GLR-unknown-post}, we use the following lemma as our concentration inequality:

\begin{lemma} \label{lem:mix_martin}
Let $\lp X_{n}\rp_{n=1}^{\infty}$ be an i.i.d. $\sigma^{2}$-sub-Gaussian sequence with mean $\mu$, then for all $\delta\in\lp0,1\rp$,
\textup{\begin{align}
    &\Pr_{\infty}\lp\exists\,n\in\mbb{N}:n\mrm{kl}\lp\hat{\mu}_{1:n},\mu\rp-3\log\lp1+\log\lp n\rp\rp>\frac{5}{4}\log\lp\frac{1}{\delta}\rp+\frac{11}{2}\rp\leq\delta.\label{eq:mix_mart}
\end{align}}
\end{lemma}

\begin{proof} [Proof of Lemma \ref{lem:mix_martin}]
Let $Y\lp n\rp\coloneqq n \mrm{kl}\lp\hat{\mu}_{1:n},\mu\rp-3\log\lp1+\log\lp n\rp\rp$. To prove this lemma, we associate the random process $\{Y \lp n \rp\}$ with a supermartingale, so that we can apply Ville's inequality. To this end, we use the following lemma derived in \cite{kaufmann2021mixture} to construct the associated supermartingale. 
\begin{lemma}[Lemma 13 from \cite{kaufmann2021mixture}]
Let $c\coloneqq\frac{5}{4}\log\lp\frac{\pi^{2}/3}{\lp\log\lp5/4\rp\rp^{2}}\rp$. For any $x > 0$, there exists a nonnegative (mixture) martingale $Z\lp t\rp$ such that $Z\lp 0 \rp=1$ and for any $x>0$ and $n\in\mbb{N}$,
\begin{equation}
    \lbp Y\lp n\rp-c\geq x\rbp\subseteq\lbp Z\lp n\rp\geq e^{\frac{4x}{5}}\rbp\label{eq:mix_martin_include-1}.
\end{equation}
\end{lemma}
Continuing with the proof of Lemma~\ref{lem:mix_martin}, for any $\lambda>0$ and $z>1$, we have:
\begin{align}\nonumber
    \lbp e^{\lambda\lb Y\lp n\rp-11/2\rb}\geq z\rbp&\overset{(a)}{\subseteq}\lbp e^{\lambda\lb Y\lp n\rp-c\rb}\geq z\rbp  \\\nonumber
    &=\lbp Y\lp n\rp - c\geq \frac{\log\lp z\rp}{\lambda}\rbp\\\nonumber
    &\overset{(b)}{\subseteq}\lbp Z\lp n\rp\geq e^{\frac{4\log\lp z\rp}{5\lambda}}\rbp\\\nonumber
    &=\lbp Z\lp n\rp\geq z^{\frac{4}{5\lambda}}\rbp\\
    &=\lbp\lp Z\lp n\rp\rp^{5\lambda/4}\geq z\rbp\label{eq:mix_martin_include-2}
\end{align}
where step $(a)$ is owing to the fact that $\frac{11}{2} \geq c$, and step $(b)$ is due to \eqref{eq:mix_martin_include-1}. When $\lambda\leq\frac{4}{5}$, because $g \lp x \rp = x^{5\lambda/4}$ is a concave function, $\lp Z\lp n\rp\rp^{5\lambda/4}$ is a supermartingale.
Hence, for any $\lambda\in\lp0,\frac{4}{5}\rb$, and $y > 11/2$, we have the following inequality:
\begin{align}\nonumber
    \Pr_{\infty}\lp\exists\,n\in\mbb{N}:Y\lp n\rp>y\rp&=\Pr_{\infty}\lp\cup_{n=1}^{\infty}\lbp Y\lp n\rp>y\rbp\rp\\\nonumber
    &=\Pr_{\infty}\lp\cup_{n=1}^{\infty}\lbp e^{\lambda Y\lp n\rp}>e^{\lambda y}\rbp\rp\\\nonumber
    &=\Pr_{\infty}\lp\cup_{n=1}^{\infty}\lbp e^{\lambda\lb Y\lp n\rp - 11/2\rb}>e^{\lambda\lp y - 11/2 \rp}\rbp\rp\\\nonumber
    &\overset{(a)}{\leq}\Pr_{\infty}\lp\cup_{n=1}^{\infty}\lbp\lp Z\lp n\rp\rp^{5\lambda/4}>e^{\lambda \lp y - 11/2 \rp}\rbp\rp\\\nonumber
    &=\Pr_{\infty}\lp\exists\,n\in\mbb{N}:\lp Z\lp n\rp\rp^{5\lambda/4}>e^{\lambda\lp y-11/2\rp}\rp\\\nonumber
    &=\Pr_{\infty}\lp \sup_{n \in \mbb{N}} \lp Z\lp n\rp\rp^{5\lambda/4}>e^{\lambda\lp y-11/2\rp}\rp\\
    &\overset{(b)}{\leq}e^{-\lambda\lp y-11/2 \rp}\label{eq:Ville}
\end{align}
where step $(a)$ is due to \eqref{eq:mix_martin_include-2} and step $(b)$ comes from Ville's inequality \cite{ville1939etude}. By minimizing \eqref{eq:Ville} over $\lambda\in\lp0,\frac{4}{5}\rb$, we obtain
\begin{align}
    \Pr_{\infty}\lp\exists\,n\in\mbb{N}:Y\lp n\rp>y\rp&\leq e^{-\frac{4}{5}\lp y-11/2 \rp}\label{eq:Minimize}.
\end{align}
Then, by letting $\delta=e^{-\frac{4}{5}\lp y-11/2 \rp}$, we can see that for any $\delta\in\lp0,1\rp$,
\begin{align}
    \Pr_{\infty}\lp\exists\,n\in\mbb{N}:Y\lp n\rp>\frac{5}{4}\log\lp\frac{1}{\delta}\rp + \frac{11}{2} \rp&\leq\delta\label{eq:lemma-end}.
\end{align}
\end{proof}

Continuing with the proof of the upper bound on the false alarm probability, for any $T\in\mbb{N}$, we can upper bound the false alarm probability as follows:
\begin{align}
    &\Pr_{\infty} \lp \tau_{\mrm{GLR}} \leq T \rp \nonumber\\
    &\leq \Pr_{\infty} \lp \tau_{\mrm{GLR}} < \infty \rp \nonumber\\
    &=\Pr_{\infty}\lp\exists\,n\in\mbb{N}: G_{n} > \beta_{\mrm{GLR}} \lp n,\delta_{\mathrm{F}}\rp\rp \nonumber\\
    &\overset{(a)}{=}\Pr_{\infty}\lp\exists\, n \in\mbb{N}:\sup_{1 \leq k \leq n} \lp n - k + 1 \rp \mathrm{kl} \lp \hat{\mu}_{k:n}; \mu_{0} \rp > \beta_{\mrm{GLR}} \lp n, \delta_{\mathrm{F}} \rp \rp \nonumber\\
    &=\Pr_{\infty}\bigg(\exists\,k,r\in\mbb{N}: r\mathrm{kl}\lp\hat{\mu}_{k:k+r-1};\mu_{0}\rp>3\log\lp1+\log\lp k+r-1\rp\rp \nonumber\\
    &\quad\quad\quad+\frac{5}{4}\log\lp\frac{3\lp k+r-1\rp^{3/2}}{\delta_{\mathrm{F}}}\rp+\frac{11}{2}\bigg) \nonumber\\
    &\overset{(b)}{\leq}\sum_{k=1}^{\infty}\Pr_{\infty}\bigg(\exists\,r\in\mbb{N}: r\mathrm{kl}\lp\hat{\mu}_{k:k+r-1};\mu_{0}\rp>3\log\lp1+\log\lp k+r-1\rp\rp \nonumber\\
    &\quad\quad\quad\quad\quad+\frac{5}{4}\log\lp\frac{3\lp k+r-1\rp^{3/2}}{\delta_{\mathrm{F}}}\rp+\frac{11}{2}\bigg) \nonumber\\
    &\overset{(c)}{\leq}\sum_{k=1}^{\infty}\Pr_{\infty}\bigg(\exists\,r\in\mbb{N}: r\mathrm{kl}\lp\hat{\mu}_{k:k+r-1};\mu_{0}\rp-3\log\lp1+\log\lp r\rp\rp \nonumber\\
    &\quad\quad\quad\quad\quad>\frac{5}{4}\log\lp\frac{3k^{3/2}}{\delta_{\mathrm{F}}}\rp+\frac{11}{2}\bigg) \nonumber\\
    &\overset{(d)}{\leq}\sum_{k=1}^{\infty}\frac{\delta_{\mathrm{F}}}{3k^{3/2}} \nonumber \\
    &\leq\delta_{\mathrm{F}}\label{eq:FA_GLR_post_unknown}.
\end{align}
In step $(a)$, we apply Lemma \ref{lem:GLR-kl-unknown-post} to represent $G_{n}$ with the empricial mean of the observations. Step $(b)$ is due to the union bound, and step $(c)$ results from the fact that $\log$ is an increasing function. In step $(d)$, we exploit the concentration inequality in Lemma \ref{lem:mix_martin}.

Next, we prove the detection delay performance of the GLR test in \eqref{def:GLR-unknown-post}: Recall the definition of $d$ in \eqref{eq:latency-unknown-post} and $\Delta = \lba \mu_{0} - \mu_{1} \rba$. For any $\nu \in \lbp 1, \dots, T - d \rbp$, we have
\begin{align}
    &\Pr_{\nu} \lp \tau_{\mrm{GLR}} \geq \nu + d \rp \nonumber\\
    &=\Pr_{\nu}\lp \forall\,n \in \lbp 1, \dots, \nu + d - 1 \rbp: G_{n} \leq \beta_{\mrm{GLR}} \lp n,\delta_{\mathrm{F}}\rp\rp \nonumber\\
    &\overset{(a)}{=} \Pr_{\nu}\lp \forall\,n \in \lbp 1, \dots, \nu + d - 1 \rbp: \sup_{1 \leq k \leq n}\lp n - k + 1 \rp \mathrm{kl} \lp \hat{\mu}_{k:n};\mu_{0}\rp \leq \beta_{\mrm{GLR}} \lp n,\delta_{\mathrm{F}}\rp\rp \nonumber\\
    &\overset{(b)}{\leq} \Pr_{\nu} \lbp \sup_{1 \leq k \leq \nu + d - 1} \lp \nu + d - k \rp \mathrm{kl} \lp \hat{\mu}_{k:\nu + d - 1}; \mu_{0} \rp \leq \beta_{\mrm{GLR}} \lp \nu + d - 1, \delta_{\mathrm{F}} \rp \rbp \nonumber\\
    &\overset{(c)}{\leq} \Pr_{\nu} \lbp d\mathrm{kl} \lp \hat{\mu}_{\nu: \nu + d - 1}; \mu_{0} \rp \leq \beta_{\mrm{GLR}} \lp \nu + d - 1, \delta_{\mathrm{F}} \rp \rbp \nonumber\\
    &\overset{(d)}{\leq} \Pr_{\nu} \lbp d \mathrm{kl} \lp \hat{\mu}_{\nu: \nu + d - 1}; \mu_{0} \rp \leq \beta_{\mrm{GLR}} \lp T, \delta_{\mrm{F}} \rp \rbp \nonumber\\
    &=\Pr_{\nu}\lbp\lba\hat{\mu}_{\nu:\nu + d - 1}-\mu_{0}\rba\leq\sqrt{\frac{2\sigma^{2}}{d} \beta_{\mrm{GLR}}\lp T,\delta_{\mathrm{F}}\rp}\rbp \nonumber\\
    &\overset{(e)}{\leq} \Pr_{\nu} \lbp \lba \hat{\mu}_{\nu: \nu + d - 1} - \mu_{1} \rba \geq \Delta - \sqrt{\frac{2 \sigma^{2}}{d} \beta_{\mrm{GLR}} \lp T, \delta_{\mathrm{F}} \rp} \rbp \label{eq:proof_Latency_GLR_post_unknown_1}
\end{align}
where step $(a)$ results from Lemma \ref{lem:GLR-kl-unknown-post}, and step $(b)$ is due to the fact that $\lbp \nu + d - 1 \rbp \subseteq \lbp 1, \dots, \nu + d - 1 \rbp$. In step $(c)$, we exploit the definition of supremum. In step $(d)$, we apply the fact that $\beta_{\mrm{GLR}}$ is increasing with the time step $n$. Step $(e)$ stems directly from triangle inequality. Next, by plugging in the definition of $d$ in \eqref{eq:latency-unknown-post}, we have
\begin{align}
    &\Pr_{\nu}\lp \tau_{\mrm{GLR}} \geq \nu + d \rp \nonumber\\
    &\overset{(a)}{\leq} \Pr_{\nu} \lbp \lba \hat{\mu}_{\nu: \nu + d - 1} - \mu_{1} \rba \geq \sqrt{\frac{2 \sigma^{2}}{d} \log \lp \frac{2}{\delta_{\mathrm{D}}} \rp} \rbp \nonumber\\
    &=\Pr_{\nu}\lbp\hat{\mu}_{\nu:\nu+d-1}-\mu_{1}\geq\sqrt{\frac{2\sigma^{2}}{d}\log\lp\frac{2}{\delta_{\mathrm{D}}}\rp}\rbp \nonumber\\
    &\quad+\Pr_{\nu}\lbp\hat{\mu}_{\nu:\nu+d-1}-\mu_{1}\geq\sqrt{\frac{2\sigma^{2}}{d}\log\lp\frac{2}{\delta_{\mathrm{D}}}\rp}\rbp \nonumber\\
    &\overset{(b)}{\leq}2\exp\lp-\frac{d}{2\sigma^{2}}\lp\sqrt{\frac{2\sigma^{2}}{d}\log\lp\frac{2}{\delta_{\mathrm{D}}}\rp}\rp^{2}\rp \nonumber\\
    &=\delta_{\mathrm{D}}\label{eq:proof_Latency_GLR_post_unknown_2}
\end{align}
where step $(a)$ stems from \eqref{eq:latency-unknown-post}. In step $(b)$, we exploit the $\sigma^{2}$-sub-Gaussianity of the post-change distribution by applying the Chernoff bound.

Now we focus on the GSR test in \eqref{def:GSR-unknown-post}. We first prove the upper bound on the false alarm probability of the GSR test using that of the GLR test in \eqref{def:GLR-unknown-post} as follows: Recall the definition of $W_{n}$ in \eqref{def:GSR-unknown-post} and $\beta_{\mrm{GSR}}$ in \eqref{def:GSR-unknown-post-thres}. For any $T \in \mbb{N}$,
\begin{align}
    &\Pr_{\infty} \lp \tau_{\mrm{GSR}} \leq T \rp \nonumber\\
    &=\Pr_{\infty} \lp \exists n \leq T: \log W_{n} \geq \beta_{\mrm{GSR}} \lp n, \delta_{\mrm{F}} \rp \rp \nonumber\\
    &=\Pr_{\infty} \lp \exists n \leq T: \sum_{k = 1}^{n}  \frac{\sup_{\mu\in\mbb{R}}\prod_{i=k}^{n}f_{\mu}(X_i)}{\prod_{i=k}^{n}f_{\mu_{0}}(X_i)} \geq n \exp \lp \beta_{\mrm{GLR}} \lp n, \delta_{\mrm{F}} \rp \rp \rp \nonumber\\
    &=\Pr_{\infty} \lp \exists n \leq T: \frac{1}{n} \sum_{k = 1}^{n}  \frac{\sup_{\mu\in\mbb{R}}\prod_{i=k}^{n}f_{\mu}(X_i)}{\prod_{i=k}^{n}f_{\mu_{0}}(X_i)} \geq \exp \lp \beta_{\mrm{GLR}} \lp n, \delta_{\mrm{F}} \rp \rp \rp \nonumber\\
    &\overset{(a)}{\leq} \Pr_{\infty} \lp \exists n \leq T: \sup_{1 \leq k \leq n}  \frac{\sup_{\mu\in\mbb{R}}\prod_{i=k}^{n}f_{\mu}(X_i)}{\prod_{i=k}^{n}f_{\mu_{0}}(X_i)} \geq \exp \lp \beta_{\mrm{GLR}} \lp n, \delta_{\mrm{F}} \rp \rp \rp \nonumber\\
    &= \Pr_{\infty} \lp \exists n \leq T: \sup_{1 \leq k \leq n} \log \lp \frac{\sup_{\mu\in\mbb{R}}\prod_{i=k}^{n}f_{\mu}(X_i)}{\prod_{i=k}^{n}f_{\mu_{0}}(X_i)} \rp \geq  \beta_{\mrm{GLR}} \lp n, \delta_{\mrm{F}} \rp \rp \nonumber\\
    &= \Pr_{\infty} \lp \exists n \leq T: G_{n} \geq  \beta_{\mrm{GLR}} \lp n, \delta_{\mrm{F}} \rp \rp \nonumber\\
    &= \Pr_{\infty} \lp \tau_{\mrm{GSR}} \leq T \rp \nonumber\\
    &\overset{(b)}{\leq} \delta_{\mrm{F}}.\label{eq:proof_FA_GSR_post_unknown}
\end{align}
In step $(a)$, we use the fact that for a set of real numbers, their mean is smaller than their supremum. In step $(b)$, we apply \eqref{eq:FA_GLR_post_unknown}.

Next, we prove the detection delay performance of the GSR test in \eqref{def:GSR-unknown-post}: Recall the definition of $d$ in \eqref{eq:latency-unknown-post} and $\Delta = \lba \mu_{0} - \mu_{1} \rba$. For any $\nu \in \lbp 1, \dots, T - d \rbp$, we have
\begin{align}
    &\Pr_{\nu} \lp \tau_{\mrm{GSR}} \geq \nu + d \rp \nonumber\\
    &=\Pr_{\nu}\lp \forall\,n \in \lbp 1, \dots, \nu + d - 1 \rbp: \log W_{n} \leq \beta_{\mrm{GSR}} \lp n,\delta_{\mathrm{F}}\rp\rp \nonumber\\
    &=\Pr_{\nu}\lp \forall\,n \in \lbp 1, \dots, \nu + d - 1 \rbp: \log \lp \sum_{k = 1}^{n}  \frac{\sup_{\mu\in\mbb{R}}\prod_{i=k}^{n}f_{\mu}(X_i)}{\prod_{i=k}^{n}f_{\mu_{0}}(X_i)} \rp \leq \beta_{\mrm{GSR}} \lp n,\delta_{\mathrm{F}}\rp\rp \nonumber\\
    &\overset{(a)}{\leq} \Pr_{\nu} \lbp \log \lp \sum_{k = 1}^{\nu + d - 1} \frac{\sup_{\mu \in \mbb{R}} \prod_{i = k}^{\nu + d - 1} f_{\mu}(X_i)}{\prod_{i = k}^{\nu + d - 1}f_{\mu_{0}}(X_i)} \rp \leq \beta_{\mrm{GSR}} \lp \nu + d - 1,\delta_{\mathrm{F}} \rp \rbp \nonumber\\
    &\leq \Pr_{\nu} \lbp \log \lp \frac{\sup_{\mu \in \mbb{R}} \prod_{i = \nu}^{\nu + d - 1} f_{\mu}(X_i)}{\prod_{i = \nu}^{\nu + d - 1}f_{\mu_{0}}(X_i)} \rp \leq \beta_{\mrm{GSR}} \lp \nu + d - 1,\delta_{\mathrm{F}} \rp \rbp \nonumber\\
    &\overset{(b)}{=} \Pr_{\nu} \lbp d\mathrm{kl} \lp \hat{\mu}_{\nu: \nu + d - 1}; \mu_{0} \rp \leq \beta_{\mrm{GSR}} \lp \nu + d - 1, \delta_{\mathrm{F}} \rp \rbp \nonumber\\
    &\overset{(c)}{\leq} \Pr_{\nu} \lbp d \mathrm{kl} \lp \hat{\mu}_{\nu: \nu + d - 1}; \mu_{0} \rp \leq \beta_{\mrm{GSR}} \lp T, \delta_{\mrm{F}} \rp \rbp \nonumber\\
    &=\Pr_{\nu}\lbp\lba\hat{\mu}_{\nu:\nu + d - 1}-\mu_{0}\rba\leq\sqrt{\frac{2\sigma^{2}}{d} \beta_{\mrm{GSR}}\lp T,\delta_{\mathrm{F}}\rp}\rbp \nonumber\\
    &\overset{(d)}{\leq} \Pr_{\nu} \lbp \lba \hat{\mu}_{\nu: \nu + d - 1} - \mu_{1} \rba \geq \Delta - \sqrt{\frac{2 \sigma^{2}}{d} \beta_{\mrm{GSR}} \lp T, \delta_{\mathrm{F}} \rp} \rbp \label{eq:proof_Latency_GSR_post_unknown}
\end{align}
where step $(a)$ stems from the fact that $\lbp \nu + d - 1 \rbp \subseteq \lbp 1, \dots, \nu + d - 1 \rbp$, and step $(b)$ is due to Lemma \ref{lem:GLR-kl-unknown-post}. In step $(c)$, we apply the fact that $\beta_{\mrm{GSR}}$ is increasing with the time step $n$. Step $(d)$ stems directly from triangle inequality. Following the same steps in \eqref{eq:proof_Latency_GLR_post_unknown_2}, we can prove that $\Pr_{\nu} \lp \tau_{\mrm{GSR}} \geq \nu + d \rp \leq \delta_{\mrm{D}}$.

\section{Proof of Theorem \ref{thm:pre-post-unknown}} \label{sec:thm2}

Similar to Theorem \ref{thm:unknown-post}, there are two tests to consider in Theorem \ref{thm:pre-post-unknown}: the GLR test in \eqref{def:GLR-unknown-pre-post} and GSR test in \eqref{def:GSR-unknown-pre-post}. For each test, there are two parts to prove: the false alarm constraint $\Pr_{\infty} \lp \tau \leq T \rp$ and the detection delay performance $\Pr_{\nu} \lp \tau \geq \nu + d \rp$. To prove these two inequalities for the GLR test, we first associate the GLR statistic with the empirical mean of the sub-Gaussian observations, so that we can exploit the sub-Gaussianity to apply the concentration inequality in Lemma \ref{lem:mix_martin}.

Recall that $\hat{\mu}_{m:n}$ is the empirical mean of samples $\lbp X_{m},\dots,X_{n}\rbp$ for any $m \leq n \in \mbb{N}$, and that $f_{\mu}$ is the density of a Gaussian random variable with mean $\mu$ and variance $\sigma^{2}$. We use the following lemma to represent the GLR statistic in \eqref{eq:GLR_statistics_pre-post-unknown} using the empirical mean of the observations. 

\begin{lemma}\label{lem:GLR-kl-unknown-pre-post}
For any $n\in\mbb{N}$ and any $k\in\lbp1,\dots,n\rbp$, we have:
\textup{\begin{align}
    &\log\lp\frac{\sup_{\mu_{0}'\in\mbb{R}}\prod_{i=1}^{k}f_{\mu_{0}'}\lp X_{i}\rp\sup_{\mu_{1}'\in\mbb{R}}\prod_{i=k+1}^{n}f_{\mu_{1}'}\lp X_{i}\rp}{\sup_{\mu\in\mbb{R}}\prod_{i=1}^{n}f_{\mu}\lp X_{i}\rp}\rp=k\mrm{kl}\lp\hat{\mu}_{1:k};\hat{\mu}_{1:n}\rp+\lp n-k\rp\mrm{kl}\lp\hat{\mu}_{k+1:n};\hat{\mu}_{1:n}\rp\label{eq:lr-kl}
\end{align}}
where \textup{$\mrm{kl}\lp x;y\rp\coloneqq\frac{(x-y)^{2}}{2\sigma^{2}}$} is the KL-divergence between two Gaussian distributions with common variance $\sigma^{2}$ and different means $x,y\in\mbb{R}$.
\end{lemma}

\begin{proof} [Proof of Lemma \ref{lem:GLR-kl-unknown-pre-post}]
We can show that for any $n\in\mbb{N}$ and any $s\in\lbp1,\dots,n\rbp$,
\begin{align}\nonumber
    &\log\lp\frac{\sup_{\theta_{0}\in\mbb{R}}\prod_{i=1}^{s}f_{\theta_{0}}\lp X_{i}\rp\sup_{\theta_{1}\in\mbb{R}}\prod_{i=s+1}^{n}f_{\theta_{1}}\lp X_{i}\rp}{\sup_{\theta\in\mbb{R}}\prod_{i=1}^{n}f_{\theta}\lp X_{i}\rp}\rp\\\nonumber
    &=\log\lp\frac{\sup_{\theta_{0}\in\mbb{R}}\prod_{i=1}^{s}\frac{1}{\sqrt{2\pi\sigma^{2}}}\exp\lp-\frac{\lp X_{i}-\theta_{0}\sigma^{2}\rp^{2}}{2\sigma^{2}}\rp\sup_{\theta_{1}\in\mbb{R}}\prod_{i=s+1}^{n}\frac{1}{\sqrt{2\pi\sigma^{2}}}\exp\lp-\frac{\lp X_{i}-\theta_{1}\sigma^{2}\rp^{2}}{2\sigma^{2}}\rp}{\sup_{\theta\in\mbb{R}}\prod_{i=1}^{n}\frac{1}{\sqrt{2\pi\sigma^{2}}}\exp\lp-\frac{\lp X_{i}-\theta\sigma^{2}\rp^{2}}{2\sigma^{2}}\rp}\rp\\\nonumber
    &=\log\lp\frac{\exp\lp-\inf_{\theta_{0}\in\mbb{R}}\sum_{i=1}^{s}\frac{\lp X_{i}-\theta_{0}\sigma^{2}\rp^{2}}{2\sigma^{2}}\rp\exp\lp-\inf_{\theta_{1}\in\mbb{R}}\sum_{i=s+1}^{n}\frac{\lp X_{i}-\theta_{1}\sigma^{2}\rp^{2}}{2\sigma^{2}}\rp}{\exp\lp-\inf_{\theta\in\mbb{R}}\sum_{i=1}^{n}\frac{\lp X_{i}-\theta\sigma^{2}\rp^{2}}{2\sigma^{2}}\rp}\rp\\\nonumber
    &\overset{(a)}{=}\log\lp\frac{\exp\lp-\sum_{i=1}^{s}\frac{\lp X_{i}-\hat{\mu}_{1:s}\rp^{2}}{2\sigma^{2}}\rp\exp\lp-\sum_{i=s+1}^{n}\frac{\lp X_{i}-\hat{\mu}_{s+1:n}\rp^{2}}{2\sigma^{2}}\rp}{\exp\lp-\sum_{i=1}^{n}\frac{\lp X_{i}-\hat{\mu}_{1:n}\rp^{2}}{2\sigma^{2}}\rp}\rp\\\nonumber
    &=\log\lp\frac{\exp\lp-\sum_{i=1}^{s}\frac{X_{i}^{2}-2X_{i}\hat{\mu}_{1:s}+\hat{\mu}_{1:s}^{2}}{2\sigma^{2}}\rp\exp\lp-\sum_{i=s+1}^{n}\frac{X_{i}^{2}-2X_{i}\hat{\mu}_{s+1:n}+\hat{\mu}_{s+1:n}^{2}}{2\sigma^{2}}\rp}{\exp\lp-\sum_{i=1}^{n}\frac{X_{i}^{2}-2X_{i}\hat{\mu}_{1:n}+\hat{\mu}_{1:n}^{2}}{2\sigma^{2}}\rp}\rp\\\nonumber
    &=\sum_{i=1}^{s}\frac{2X_{i}\hat{\mu}_{1:s}-\hat{\mu}_{1:s}^{2}}{2\sigma^{2}}+\sum_{i=s+1}^{n}\frac{2X_{i}\hat{\mu}_{s+1:n}-\hat{\mu}_{s+1:n}^{2}}{2\sigma^{2}}-\sum_{i=1}^{n}\frac{2X_{i}\hat{\mu}_{1:n}-\hat{\mu}_{1:n}^{2}}{2\sigma^{2}}\\\nonumber
    &=s\frac{\hat{\mu}_{1:s}^{2}}{2\sigma^{2}}+\lp n-s\rp\frac{\hat{\mu}_{s+1:n}^{2}}{2\sigma^{2}}-n\frac{\hat{\mu}_{1:n}^{2}}{2\sigma^{2}}\\\nonumber
    &=s\frac{\hat{\mu}_{1:s}^{2}}{2\sigma^{2}}+\lp n-s\rp\frac{\hat{\mu}_{s+1:n}^{2}}{2\sigma^{2}}+s\frac{\hat{\mu}_{1:n}^{2}}{2\sigma^{2}}+\lp n-s\rp\frac{\hat{\mu}_{1:n}^{2}}{2\sigma^{2}}-2\frac{s\hat{\mu}_{1:s}\hat{\mu}_{1:n}}{2\sigma^{2}}-2\frac{\lp n-s\rp\hat{\mu}_{s+1:n}\hat{\mu}_{1:n}}{2\sigma^{2}}\\\nonumber
    &=s\frac{\lp\hat{\mu}_{1:s}-\hat{\mu}_{1:n}\rp^{2}}{2\sigma^{2}}+\lp n-s\rp\frac{\lp\hat{\mu}_{s+1:n}-\hat{\mu}_{1:n}\rp^{2}}{2\sigma^{2}}\\
    &=s\mrm{kl}\lp\hat{\mu}_{1:s};\hat{\mu}_{1:n}\rp+\lp n-s\rp\mrm{kl}\lp\hat{\mu}_{s+1:n};\hat{\mu}_{1:n}\rp\label{eq:lem1_proof}
\end{align}
where step $(a)$ follows from the fact that $\sum_{i=t}^{t'}\lp X_{i}-a\rp^{2}$ is minimized when $a=\hat{\mu}_{t:t'}$.
\end{proof}

Recall that $\mu_{0}$ is the pre-change mean of the sample sequence $\lp X_{i} \rp_{i = 1}^{\infty}$. By Lemmas \ref{lem:mix_martin} and \ref{lem:GLR-kl-unknown-pre-post}, for any $T \in \mbb{N}$,
\begin{align}\nonumber
    &\Pr_{\infty}\lp \tilde{\tau}_{\mrm{GLR}} \leq T\rp\\\nonumber
    &\leq \Pr_{\infty}\lp \tilde{\tau}_{\mrm{GLR}}<\infty\rp\\\nonumber
    &=\Pr_{\infty}\lp\exists\,n\in\mbb{N}:\;\sup_{1\leq k\leq n}\log\lp\frac{\sup_{\mu_{0}'\in\mbb{R}}\prod_{i=1}^{k}f_{\mu_{0}'}\lp X_{i}\rp\sup_{\mu_{1}'\in\mbb{R}}\prod_{i=k+1}^{n}f_{\mu_{1}'}\lp X_{i}\rp}{\sup_{\mu\in\mbb{R}}\prod_{i=1}^{n}f_{\mu}\lp X_{i}\rp}\rp\geq \tilde{\beta}_{\mrm{GLR}} \lp n,\delta_{\mrm{F}}\rp\rp\\\nonumber
    &\overset{(a)}{=}\Pr_{\infty}\lp\exists\,n\in\mbb{N}:\;\sup_{1\leq k\leq n}k\mrm{kl}\lp\hat{\mu}_{1:k},\hat{\mu}_{1:n}\rp+\lp n-k\rp\mrm{kl}\lp\hat{\mu}_{k+1:n},\hat{\mu}_{1:n}\rp\geq\tilde{\beta}_{\mrm{GLR}} \lp n,\delta_{\mrm{F}}\rp\rp\\\nonumber
    &=\Pr_{\infty}\left(\exists\,k\leq n\in\mathbb{N}:k\mrm{kl}(\hat{\mu}_{1:k},\hat{\mu}_{1:n})+(n-k)\mrm{kl}(\hat{\mu}_{k+1:n},\hat{\mu}_{1:n})> \tilde{\beta}_{\mrm{GLR}} \lp n,\delta_{\mrm{F}}\rp\right)\\\nonumber
    &\overset{(b)}{=}\Pr_{\infty}\left(\exists\,k\leq n\in\mbb{N}:\inf_{\mu}k\mrm{kl}(\hat{\mu}_{1:k},\mu)+(n-k)\mrm{kl}(\hat{\mu}_{k+1:n},\mu)> \tilde{\beta}_{\mrm{GLR}}(n,\delta_{\mrm{F}})\right)\\\nonumber
    &\leq\Pr_{\infty}\left(\exists\,k\leq n\in\mathbb{N}:k\mrm{kl}(\hat{\mu}_{1:k},\mu_{0})+(n-k)\mrm{kl}(\hat{\mu}_{k+1:n},\mu_{0})> \tilde{\beta}_{\mrm{GLR}}(n,\delta_{\mrm{F}})\right)\\\nonumber
    &=\Pr_{\infty}\Bigg(\exists\,k,r\in\mathbb{N}:\\\nonumber
    &\quad\quad\quad\quad s\mrm{kl}(\hat{\mu}_{1:k},\mu_{0})+r\mrm{kl}(\hat{\mu}_{k+1:k+r},\mu_{0})>6\log(1+\log(k+r))+\frac{5}{2}\log\lp\frac{4\lp k+r\rp^{\frac{3}{2}}}{\delta_{\mrm{F}}}\rp+11\Bigg)\\\nonumber
    &\overset{(c)}{\leq}\Pr_{\infty}\Bigg(\exists\,k,r\in\mathbb{N}:\lbp k\mrm{kl}(\hat{\mu}_{1:k},\mu_{0})>3\log(1+\log(k+r))+\frac{5}{4}\log\lp\frac{4\lp k+r\rp^{3/2}}{\delta_{\mrm{F}}}\rp+\frac{11}{2}\rbp\cup\\\nonumber
    &\quad\quad\quad\Bigg\{ r\mrm{kl}(\hat{\mu}_{k+1:k+r},\mu_{0}) > 3\log(1+\log(k+r)) +\frac{5}{4}\log\lp\frac{4\lp k+r\rp^{3/2}}{\delta_{\mrm{F}}}\rp+\frac{11}{2} \Bigg\} \Bigg)\\\nonumber
    &=\Pr_{\infty}\Bigg(\Bigg\{\exists\,k,r\in\mathbb{N}:k\mrm{kl}(\hat{\mu}_{1:k},\mu_{0})>3\log(1+\log(k+r))+\frac{5}{4}\log\lp\frac{4\lp k+r\rp^{3/2}}{\delta_{\mrm{F}}}\rp+\frac{11}{2}\Bigg\}\cup\\\nonumber
    &\quad\quad\quad\Bigg\{\exists\,k,r\in\mathbb{N}:r\mrm{kl}(\hat{\mu}_{k+1:k+r},\mu_{0})>3\log(1+\log(k+r))+\frac{5}{4}\log\lp\frac{4\lp k+r\rp^{3/2}}{\delta_{\mrm{F}}}\rp+\frac{11}{2}\Bigg\}\Bigg)\\\nonumber
    &\overset{(d)}{\leq}\Pr_{\infty}\bigg(\left\{\exists\,k\in\mathbb{N}:k\mrm{kl}(\hat{\mu}_{1:k},\mu_{0})>3\log(1+\log(k))+\frac{5}{4}\log\lp\frac{4}{\delta_{\mrm{F}}}\rp+\frac{11}{2}\right\}\cup\\\nonumber
    &\quad\quad\quad\;\bigg\{\exists\,k,r\in\mathbb{N}:r\mrm{kl}(\hat{\mu}_{k+1:k+r},\mu_{0})>3\log(1+\log(r))+\frac{5}{4}\log\lp\frac{4k^{3/2}}{\delta_{\mrm{F}}}\rp+\frac{11}{2}\bigg\}\bigg) \\\nonumber
    &\overset{(e)}{\leq}\Pr_{\infty}\left(\exists\,k\in\mathbb{N}:k\mrm{kl}(\hat{\mu}_{1:k},\mu_{0})>3\log(1+\log(k))+\frac{5}{4}\log\lp\frac{4}{\delta_{\mrm{F}}}\rp+\frac{11}{2}\right)+\\\nonumber
    &\quad\;\Pr_{\infty}\left(\exists\,k,r\in\mathbb{N}:r\mrm{kl}(\hat{\mu}_{k+1:k+r},\mu_{0})>3\log(1+\log(r))+\frac{5}{4}\log\lp\frac{4k^{3/2}}{\delta_{\mrm{F}}}\rp+\frac{11}{2}\right)\\\nonumber
    &=\Pr_{\infty}\left(\exists\,k\in\mathbb{N}:k\mrm{kl}(\hat{\mu}_{1:k},\mu_{0})-3\log(1+\log(k))>\frac{5}{4}\log\lp\frac{4}{\delta_{\mrm{F}}}\rp+\frac{11}{2}\right)+\\\nonumber
    &\quad\;\Pr_{\infty}\bigg(\bigcup_{k=1}^{\infty}\bigg\{\exists\,r\in\mbb{N}:r\mrm{kl}(\hat{\mu}_{k+1:k+r},\mu_{0})-3\log(1+\log(r))>\frac{5}{4}\log\lp\frac{4k^{3/2}}{\delta_{\mrm{F}}}\rp + \frac{11}{2}\bigg\}\bigg)\\\nonumber
    &\overset{(f)}{\leq}\Pr_{\infty}\left(\exists\,k\in\mathbb{N}:k\mrm{kl}(\hat{\mu}_{1:k},\mu_{0})-3\log(1+\log(k))>\frac{5}{4}\log\lp\frac{4}{\delta_{\mrm{F}}}\rp+\frac{11}{2}\right)+\\\nonumber
    &\quad\;\sum_{k=1}^{\infty}\Pr_{\infty}\bigg(\exists\,r\in\mbb{N}:r\mrm{kl}(\hat{\mu}_{k+1:k+r},\mu_{0})-3\log(1+\log(r))>\frac{5}{4}\log\lp\frac{4k^{3/2}}{\delta_{\mrm{F}}}\rp+\frac{11}{2}\bigg)\\\nonumber
    &\overset{(g)}{\leq}\frac{\delta_{\mrm{F}}}{4}+\sum_{k=1}^{\infty}\frac{\delta_{\mrm{F}}}{4k^{3/2}}\\
    &\leq\delta_{\mrm{F}}\label{eq:Prop_2_proof}
\end{align}
where step $(a)$ is due to Lemma \ref{lem:GLR-kl-unknown-pre-post} and step $(b)$ is owing to the fact that $\inf_{\mu} k \lp \hat{\mu}_{1:k} - \mu \rp^{2} + \lp n - k \rp \lp \hat{\mu}_{k + 1 : n} - \mu \rp^{2} = k \lp \hat{\mu}_{1:k} - \hat{\mu}_{1:n} \rp^{2} + \lp n - k \rp \lp \hat{\mu}_{k + 1 : n} - \hat{\mu}_{1:n} \rp^{2}$. Step $(c)$ is due to the fact that $x+y>2a$ implies $x>a$ or $y>a$. Step $(d)$ stems from the fact that $\beta\lp n,\delta\rp$ is increasing with $n$. Steps $(e)$ and $(f)$ are owing to the union bound. By Lemma \ref{lem:mix_martin}, we obtain step $(g)$. This completes the proof of the false alarm constraint in Theorem \ref{thm:pre-post-unknown}.

We now move on to proving the detection delay performance $\Pr_{\nu} \lp \tilde{\tau}_{\mrm{GLR}} \geq \nu + d \rp$ in Theorem \ref{thm:pre-post-unknown}. To this end, we use the following lemma borrowed from \cite{besson2022efficient} as our concentration inequality:

\begin{lemma}[Lemma 10 in \cite{besson2022efficient}] \label{lem:sub_Gaussian_diff}
    Let $\hat{\mu}_{i,s}$ be the empirical mean of $s$ i.i.d. $\sigma^{2}$-sub-Gaussian samples with mean $\mu_{i}$, $i\in\lbp0,1\rbp$. Then, $\forall\,s,r\in\mbb{N}$, we have

    \begin{equation}
        \Pr\lp\frac{sr}{s+r}\lp\lp\hat{\mu}_{0,s}-\hat{\mu}_{1,r}\rp-\lp\mu_{0}-\mu_{1}\rp\rp^{2}>u\rp\leq2\exp\lp-\frac{u}{2\sigma^{2}}\rp.
    \end{equation}
\end{lemma}

Continuing with the proof of the latency, for convenience in notation,  using Lemma \ref{lem:GLR-kl-unknown-pre-post}, we can show that for any $T\in\mbb{N},\;\delta_{\mrm{D}},\delta_{\mrm{F}}\in\lp0,1\rp,\;\Delta>0$, $m>\frac{8\sigma^{2}}{\Delta^{2}}\tilde{\beta}_{\mrm{GLR}}\lp T,\delta_{\mrm{F}}\rp$, and $\nu\in\lbp m+1,\dots,T-d\rbp$, we have
\begin{align}\nonumber
    &\Pr_{\nu}\lp \tilde{\tau}_{\mrm{GLR}}\geq\nu+d\rp\\\nonumber
    &=\Pr_{\nu}\bigg(\forall\,n\in\lbp1,\dots,\nu+d-1\rbp:\;\sup_{1\leq k\leq n}\log\lp\frac{\sup_{\mu_{0}'\in\mbb{R}}\prod_{i=1}^{k}f_{\mu_{0}'}\lp X_{i}\rp\sup_{\mu_{1}'\in\mbb{R}}\prod_{i=k+1}^{n}f_{\mu_{1}'}\lp X_{i}\rp}{\sup_{\mu\in\mbb{R}}\prod_{i=1}^{n}f_{\mu}\lp X_{i}\rp}\rp\\\nonumber
    &\quad\quad\quad<\tilde{\beta}_{\mrm{GLR}}\lp n,\delta_{\mrm{F}}\rp\bigg)\\\nonumber
    &\overset{(a)}{=}\Pr_{\nu}\lp\forall\,n\in\lbp1,\dots,\nu+d-1\rbp:\;\sup_{1\leq k\leq n}k\mrm{kl}\lp\hat{\mu}_{1:k},\hat{\mu}_{1:n}\rp+\lp n-k\rp\mrm{kl}\lp\hat{\mu}_{k+1:n},\hat{\mu}_{1:n}\rp<\tilde{\beta}_{\mrm{GLR}}\lp n,\delta_{\mrm{F}}\rp\rp\\\nonumber
    &\overset{(b)}{\leq}\Pr\bigg(\sup_{1\leq k\leq \nu+d-1} k\mrm{kl}\lp\hat{\mu}_{1:k},\hat{\mu}_{1:\nu+d-1}\rp+ \lp\nu+d-1-k\rp\mrm{kl}\lp\hat{\mu}_{k+1:\nu+d-1},\hat{\mu}_{1:\nu+d-1}\rp\\\nonumber
    &\quad\quad\quad<\tilde{\beta}_{\mrm{GLR}} \lp \nu+d-1,\delta_{\mrm{F}}\rp\bigg)\\\nonumber
    &\leq \Pr_{\nu}\lp\lp\nu-1\rp\mrm{kl}\lp\hat{\mu}_{1:\nu-1},\hat{\mu}_{1:\nu+d-1}\rp+d\mrm{kl}\lp\hat{\mu}_{\nu:\nu+d-1},\hat{\mu}_{1:\nu+d-1}\rp<\tilde{\beta}_{\mrm{GLR}}\lp\nu+d-1,\delta_{\mrm{F}}\rp\rp\\\nonumber
    &=\Pr_{\nu}\Bigg(\frac{\nu-1}{2\sigma^{2}}\lp\hat{\mu}_{1:\nu-1}-\frac{\lp\nu-1\rp\hat{\mu}_{1:\nu-1}+d\hat{\mu}_{\nu:\nu+d-1}}{\nu+d-1}\rp^{2}\\\nonumber
    &\quad\quad\quad+\frac{d}{2\sigma^{2}}\lp\hat{\mu}_{\nu:\nu+d-1}-\frac{\lp\nu-1\rp\hat{\mu}_{1:\nu-1}+d\hat{\mu}_{\nu:\nu+d-1}}{\nu+d-1}\rp^{2}<\tilde{\beta}_{\mrm{GLR}}\lp\nu+d-1,\delta_{\mrm{F}}\rp\Bigg)\\
    &=\Pr_{\nu}\lp\frac{\lp\nu-1\rp d}{2\sigma^{2}\lp\nu+d-1\rp}\lp\hat{\mu}_{1:\nu-1}-\hat{\mu}_{\nu:\nu+d-1}\rp^{2}<\tilde{\beta}_{\mrm{GLR}}\lp\nu+d-1,\delta_{\mrm{F}}\rp\rp\label{eq:glr-ld-1}
\end{align}
where step $(a)$ comes from Lemma \ref{lem:GLR-kl-unknown-pre-post} and step $(b)$ results from $\lbp\nu+d-1\rbp\subseteq\lbp1,\dots,\nu+d-1\rbp$. 

Recall that $\mu_{0}$ and $\mu_{1}$ are the pre- and post-change means, and that the definition of $d$ is given in \eqref{eq:d}. For applying Lemma \ref{lem:sub_Gaussian_diff}, we need to convert $\lp\hat{\mu}_{1:\nu-1}-\hat{\mu}_{\nu:\nu+d-1}\rp^{2}$ in the last line of \eqref{eq:glr-ld-1} into $\lp\lp\hat{\mu}_{1:\nu-1}-\hat{\mu}_{\nu:\nu+d-1}\rp-\lp\mu_{0}-\mu_{1}\rp\rp^{2}$. To this end, we show that for any $\nu\in\lbp m+1,\dots,T-d\rbp$,   $\lbp\frac{\lp\nu-1\rp d}{2\sigma^{2}\lp\nu+d-1\rp}\lp\hat{\mu}_{1:\nu-1}-\hat{\mu}_{\nu:\nu+d-1}\rp^{2}<\tilde{\beta}_{\mrm{GLR}}\lp\nu+d-1,\delta_{\mrm{F}}\rp\rbp$ implies  $\lbp\frac{\lp\nu-1\rp d}{2\sigma^{2}\lp\nu+d-1\rp}\lp\lp\hat{\mu}_{1:\nu-1}-\hat{\mu}_{\nu:\nu+d-1}\rp-\lp\mu_{0}-\mu_{1}\rp\rp^{2}\geq\tilde{\beta}_{\mrm{GLR}}\lp\nu+d-1,\delta_{\mrm{F}}\rp\rbp$ with the choice of $m$ in \eqref{eq:m} and $d$ in \eqref{eq:d}. Let $\beta = \tilde{\beta}_{\mrm{GLR}}$ for notational convenience. We can show this implication as follows:
\begin{align}\nonumber
    &\lbp\frac{\lp\nu-1\rp d}{2\sigma^{2}\lp\nu+d-1\rp}\lp\hat{\mu}_{1:\nu-1}-\hat{\mu}_{\nu:\nu+d-1}\rp^{2}<\beta\lp\nu+d-1,\delta_{\mrm{F}}\rp\rbp\\\nonumber
    &\quad\cap\lbp\frac{\lp\nu-1\rp d}{2\sigma^{2}\lp\nu+d-1\rp}\lp\lp\hat{\mu}_{1:\nu-1}-\hat{\mu}_{\nu:\nu+d-1}\rp-\lp\mu_{0}-\mu_{1}\rp\rp^{2}<\beta\lp\nu+d-1,\delta_{\mrm{F}}\rp\rbp\\\nonumber
    &=\lbp\lba\hat{\mu}_{1:\nu-1}-\hat{\mu}_{\nu:\nu+d-1}\rba<\lp\frac{2\sigma^{2}\lp\nu+d-1\rp}{\lp\nu-1\rp d}\beta\lp\nu+d-1,\delta_{\mrm{F}}\rp\rp^{\frac{1}{2}}\rbp\\\nonumber
    &\quad\cap\lbp\lba\lp\hat{\mu}_{1:\nu-1}-\hat{\mu}_{\nu:\nu+d-1}\rp-\lp\mu_{0}-\mu_{1}\rp\rba<\lp\frac{2\sigma^{2}\lp\nu+d-1\rp}{\lp\nu-1\rp d}\beta\lp\nu+d-1,\delta_{\mrm{F}}\rp\rp^{\frac{1}{2}}\rbp\\\nonumber
    &\overset{(a)}{\subseteq}\lbp\lba\hat{\mu}_{1:\nu-1}-\hat{\mu}_{\nu:\nu+d-1}\rba<\lp\frac{2\sigma^{2}\lp\nu+d-1\rp}{\lp\nu-1\rp d}\beta\lp\nu+d-1,\delta_{\mrm{F}}\rp\rp^{\frac{1}{2}}\rbp\\\nonumber
    &\quad\cap\lbp\lba\mu_{0}-\mu_{1}\rba-\lba\hat{\mu}_{1:\nu-1}-\hat{\mu}_{\nu:\nu+d-1}\rba<\lp\frac{2\sigma^{2}\lp\nu+d-1\rp}{\lp\nu-1\rp d}\beta\lp\nu+d-1,\delta_{\mrm{F}}\rp\rp^{\frac{1}{2}}\rbp\\\nonumber
    &=\lbp\lba\hat{\mu}_{1:\nu-1}-\hat{\mu}_{\nu:\nu+d-1}\rba<\lp\frac{2\sigma^{2}\lp\nu+d-1\rp}{\lp\nu-1\rp d}\beta\lp\nu+d-1,\delta_{\mrm{F}}\rp\rp^{\frac{1}{2}}\rbp\\\nonumber
    &\quad\cap\lbp\lba\hat{\mu}_{1:\nu-1}-\hat{\mu}_{\nu:\nu+d-1}\rba>\Delta-\lp\frac{2\sigma^{2}\lp\nu+d-1\rp}{\lp\nu-1\rp d}\beta\lp\nu+d-1,\delta_{\mrm{F}}\rp\rp^{\frac{1}{2}}\rbp\\\nonumber
    &\subseteq\lbp\Delta<2\lp\frac{2\sigma^{2}\lp\nu+d-1\rp}{\lp\nu-1\rp d}\beta\lp\nu+d-1,\delta_{\mrm{F}}\rp\rp^{\frac{1}{2}}\rbp\\\nonumber
    &=\lbp\Delta^{2}<8\sigma^{2}\lp\frac{1}{\nu-1}+\frac{1}{d}\rp\beta\lp\nu+d-1,\delta_{\mrm{F}}\rp\rbp\\\nonumber
    &\overset{(b)}{\subseteq}\lbp\Delta^{2}<8\sigma^{2}\lp\frac{1}{m}+\frac{1}{d}\rp\beta\lp T,\delta_{\mrm{F}}\rp\rbp\\\nonumber
    &=\lbp\lp\frac{\Delta^{2}}{8\sigma^{2}\beta\lp T,\delta_{\mrm{F}}\rp}-\frac{1}{m}\rp^{-1}>d\rbp\\\nonumber
    &=\lbp \frac{8\sigma^{2}m\beta\lp T,\delta_{\mrm{F}}\rp}{\Delta^{2}m-8\sigma^{2}\beta\lp T,\delta_{\mrm{F}}\rp}>\lce\max\lbp\frac{8\sigma^{2}m\beta\lp T,\delta_{\mrm{F}}\rp}{\Delta^{2}m-8\sigma^{2}\beta\lp T,\delta_{\mrm{F}}\rp},\frac{\delta_{\mrm{F}}^{2/3}}{2^{16/15}\delta_{\mrm{D}}^{4/15}}-m\rbp\rce\rbp\\
    &=\emptyset\label{eq:implication_delay}
\end{align}
where step $(a)$ is due to triangle inequality and step $(b)$ is due to the fact that $\nu\geq m+1$ and $\nu\leq T-d$. Hence, the late detection probability can be bounded using Lemma \ref{lem:sub_Gaussian_diff} and we obtain that for any $\nu\in\lbp m+1,\dots,T-d\rbp$
\begin{align}\nonumber
    &\Pr_{\nu}\lp\tilde{\tau}_{\mrm{GLR}}\geq\nu+d\rp\\\nonumber
    &\leq\Pr_{\nu}\lp\frac{\lp\nu-1\rp d}{2\sigma^{2}\lp\nu+d-1\rp}\lp\lp\hat{\mu}_{1:\nu-1}-\hat{\mu}_{\nu:\nu+d-1}\rp-\lp\mu_{0}-\mu_{1}\rp\rp^{2}\geq\tilde{\beta}_{\mrm{GLR}}\lp\nu+d-1,\delta_{\mrm{F}}\rp\rp\\\nonumber
    &\overset{(a)}{\leq}\Pr_{\nu}\lp\frac{\lp\nu-1\rp d}{2\sigma^{2}\lp\nu+d-1\rp}\lp\lp\hat{\mu}_{1:\nu-1}-\hat{\mu}_{\nu:\nu+d-1}\rp-\lp\mu_{0}-\mu_{1}\rp\rp^{2}\geq\tilde{\beta}_{\mrm{GLR}}\lp m+d,\delta_{\mrm{F}}\rp\rp\\\nonumber
    &\overset{(b)}{\leq}\Pr_{\nu}\lp\frac{\lp\nu-1\rp d}{2\sigma^{2}\lp\nu+d-1\rp}\lp\lp\hat{\mu}_{1:\nu-1}-\hat{\mu}_{\nu:\nu+d-1}\rp-\lp\mu_{0}-\mu_{1}\rp\rp^{2}\geq\frac{5}{2}\log\lp\frac{4\lp m+d\rp^{3/2}}{\delta_{\mrm{F}}}\rp\rp\\\nonumber
    &=\Pr_{\nu}\lp\frac{\lp\nu-1\rp d}{\nu+d-1}\lp\lp\hat{\mu}_{1:\nu-1}-\hat{\mu}_{\nu:\nu+d-1}\rp-\lp\mu_{0}-\mu_{1}\rp\rp^{2}\geq2\sigma^{2}\log\lp\frac{32\lp m+d \rp^{15/4}}{\delta_{\mrm{F}}^{5/2}}\rp\rp\\\nonumber
    &\overset{(c)}{\leq}\frac{\delta_{\mrm{F}}^{5/2}}{16\lp m+d\rp^{15/4}}\\\nonumber
    &\leq\frac{\delta_{\mrm{F}}^{5/2}}{16\lp\delta_{\mrm{F}}^{5/2}2^{-16/15}\delta^{-4/15}_{\mrm{D}}\rp^{15/4}}\\
    &=\delta_{\mrm{D}}
\end{align}
where step $(a)$ is due to the fact that $\beta\lp n,\delta_{\mrm{F}}\rp$ is increasing with $n$, whereas step $(b)$ is owing to the fact that $\tilde{\beta}_{\mrm{GLR}} \lp n,\delta_{\mrm{F}}\rp\geq\frac{5}{2}\log\lp4n^{3/2}/\delta_{\mrm{F}}\rp$. Step $(c)$ comes from Lemma \ref{lem:sub_Gaussian_diff}.

Next, we prove the results for the GSR test: To prove the upper bound on false alarm probability, we use our results for the false alarm probability of the GLR test \eqref{eq:Prop_2_proof}. Similarly, the proof for the upper bound on the late detection probability $\Pr_{\nu} \lp \tilde{\tau}_{\mrm{GSR}} \geq \nu + d \rp$ uses the concentration inequality in Lemma \ref{lem:sub_Gaussian_diff}, and the steps are analogous to those used in the proof for $\Pr_{\nu} \lp \tilde{\tau}_{\mrm{GLR}} \geq \nu + d \rp \leq \delta_{\mrm{D}}$. 

We first prove $\Pr_{\infty}\lp\tilde{\tau}_{\mrm{GSR}} \leq T \rp$ using \eqref{eq:Prop_2_proof} as follows:
\begin{align}\nonumber
    &\Pr_{\infty}\lp\tilde{\tau}_{\mrm{GSR}} \leq T \rp\\\nonumber
    &=\Pr_{\infty}\lp\exists\,n\in\lbp 1, \dots, T \rbp:\;\log \tilde{W}_{n}\geq\tilde{\beta}_{\mrm{GLR}}\lp n,\delta_{\mrm{F}}\rp + \log n \rp\\\nonumber
    &=\Pr_{\infty}\Bigg(\exists\,n\in\lbp 1, \dots, T \rbp:\\\nonumber
    &\quad\quad\quad\;\sum_{k=1}^{n}\frac{\sup_{\mu_{0}'\in\mbb{R}}\prod_{i=1}^{k}f_{\mu_{0}'}\lp X_{i}\rp\sup_{\mu_{1}'\in\mbb{R}}\prod_{i=k+1}^{n}f_{\mu_{1}'}\lp X_{i}\rp}{\sup_{\mu\in\mbb{R}}\prod_{i=1}^{n}f_{\mu}\lp X_{i}\rp}\geq n\exp\lp\tilde{\beta}_{\mrm{GLR}}\lp n,\delta_{\mrm{F}}\rp\rp\Bigg)\\\nonumber
    &=\Pr_{\infty}\bigg(\exists\,n\in\lbp 1, \dots, T \rbp:\\\nonumber
    &\quad\quad\quad\;\frac{1}{n}\sum_{k=1}^{n}\frac{\sup_{\mu_{0}'\in\mbb{R}}\prod_{i=1}^{k}f_{\mu_{0}'}\lp X_{i}\rp\sup_{\mu_{1}'\in\mbb{R}}\prod_{i=k+1}^{n}f_{\mu_{1}'}\lp X_{i}\rp}{\sup_{\mu\in\mbb{R}}\prod_{i=1}^{n}f_{\mu}\lp X_{i}\rp}\geq\exp\lp\tilde{\beta}_{\mrm{GLR}}\lp n,\delta_{\mrm{F}}\rp\rp\bigg)\\\nonumber
    &\overset{(a)}{\leq}\Pr_{\infty}\bigg(\exists\,n\in\lbp 1, \dots, T \rbp:\\\nonumber
    &\quad\quad\quad\;\sup_{1\leq k\leq n}\frac{\sup_{\mu_{0}'\in\mbb{R}}\prod_{i=1}^{k}f_{\mu_{0}'}\lp X_{i}\rp\sup_{\mu_{1}'\in\mbb{R}}\prod_{i=k+1}^{n}f_{\mu_{1}'}\lp X_{i}\rp}{\sup_{\mu\in\mbb{R}}\prod_{i=1}^{n}f_{\mu}\lp X_{i}\rp}\geq\exp\lp\tilde{\beta}_{\mrm{GLR}}\lp n,\delta_{\mrm{F}}\rp\rp\bigg)\\\nonumber
    &=\Pr_{\infty}\bigg(\exists\,n\in\lbp 1, \dots, T \rbp:\\\nonumber
    &\quad\quad\quad\;\sup_{1\leq k\leq n}\log\lp\frac{\sup_{\mu_{0}'\in\mbb{R}}\prod_{i=1}^{k}f_{\mu_{0}'}\lp X_{i}\rp\sup_{\mu_{1}'\in\mbb{R}}\prod_{i=k+1}^{n}f_{\mu_{1}'}\lp X_{i}\rp}{\sup_{\mu\in\mbb{R}}\prod_{i=1}^{n}f_{\mu}\lp X_{i}\rp}\rp\geq\tilde{\beta}_{\mrm{GLR}}\lp n,\delta_{\mrm{F}}\rp\bigg)\\\nonumber
    &=\Pr_{\infty}\lp\exists\,n\in\lbp 1, \dots, T \rbp:\; \tilde{G}_{n}\geq\tilde{\beta}_{\mrm{GLR}}\lp n,\delta_{\mrm{F}}\rp\rp\\\nonumber
    &=\Pr_{\infty}\lp \tilde{\tau}_{\mrm{GLR}} \leq T \rp\\
    &\overset{(b)}{\leq}\delta_{\mrm{F}}\label{eq:gsr-fa-1}
\end{align}
where step $(a)$ results from the fact that $\sup_{1\leq i\leq n}x_{i}\leq a$ implies $\sum_{i=1}^{n}x_{i}\leq na$, whereas $(b)$ stems from \eqref{eq:Prop_2_proof}. 

Next, we prove the upper bound on the late detection probability $\Pr_{\nu}\lp\tilde{\tau}_{\mrm{GSR}}\geq\nu+d\rp$: For any $T\in\mbb{N},\;\,\delta_{\mrm{D}},\delta_{\mrm{F}}\in\lp0,1\rp,\;\Delta>0$, $m>\frac{8\sigma^{2}}{\Delta^{2}}\tilde{\beta}_{\mrm{GSR}}\lp T,\delta_{\mrm{F}}\rp$, and $\nu\in\lbp m+1,\dots,T-d\rbp$, we have
\begin{align}\nonumber
    &\Pr_{\nu}\lp\tilde{\tau}_{\mrm{GSR}}\geq\nu+d\rp\\\nonumber
    &=\Pr_{\nu}\lp\forall\,n<\nu+ d:\;\log \tilde{W}_{n}<\tilde{\beta}_{\mrm{GSR}}\lp n,\delta_{\mrm{F}}\rp \rp\\\nonumber
    &=\Pr_{\nu}\lp\forall\,n<\nu+ d:\;\log\lp\sum_{k=1}^{n}\frac{\sup_{\mu_{0}'\in\mbb{R}}\prod_{i=1}^{k}f_{\mu_{0}'}\lp X_{i}\rp\sup_{\mu_{1}'\in\mbb{R}}\prod_{i=k+1}^{n}f_{\mu_{1}'}\lp X_{i}\rp}{\sup_{\mu\in\mbb{R}}\prod_{i=1}^{n}f_{\mu}\lp X_{i}\rp}\rp<\tilde{\beta}_{\mrm{GSR}}\lp n,\delta_{\mrm{F}}\rp\rp\\\nonumber
    &\overset{(a)}{=}\Pr_{\nu}\Bigg(\forall\,n\in\lbp1,\dots,\nu+ d-1\rbp:\;\log\lp\sum_{k=1}^{n}\exp\lp k\mrm{kl}\lp\hat{\mu}_{1:k},\hat{\mu}_{1:n}\rp+\lp n-k\rp\mrm{kl}\lp\hat{\mu}_{k+1:n},\hat{\mu}_{1:n}\rp\rp\rp\\\nonumber
    &\quad\quad\quad<\tilde{\beta}_{\mrm{GSR}}\lp n,\delta_{\mrm{F}}\rp\Bigg)\\\nonumber
    &\overset{(b)}{\leq}\Pr_{\nu}\Bigg(\log\lp\sum_{k=1}^{\nu+ d - 1} \exp\lp k\mrm{kl}\lp\hat{\mu}_{1:k},\hat{\mu}_{1: \nu + d - 1} \rp + \lp \nu + d - 1 - k\rp\mrm{kl}\lp\hat{\mu}_{k+1:\nu+ d-1},\hat{\mu}_{1:\nu+ d-1}\rp\rp\rp \\\nonumber
    &\quad\quad\quad<\tilde{\beta}_{\mrm{GSR}}\lp \nu + d - 1, \delta_{\mrm{F}}\rp\Bigg)\\\nonumber
    &\leq\Pr_{\nu}\lp\lp\nu-1\rp\mrm{kl}\lp\hat{\mu}_{1:\nu-1},\hat{\mu}_{1:\nu+ d-1}\rp+ d\mrm{kl}\lp\hat{\mu}_{\nu:\nu+ d-1},\hat{\mu}_{1:\nu+ d-1}\rp<\tilde{\beta}_{\mrm{GSR}}\lp\nu+ d-1,\delta_{\mrm{F}}\rp\rp\\\nonumber
    &=\Pr_{\nu}\Bigg(\frac{\nu-1}{2\sigma^{2}}\lp\hat{\mu}_{1:\nu-1}-\frac{\lp\nu-1\rp\hat{\mu}_{1:\nu-1}+ d\hat{\mu}_{\nu:\nu+ d-1}}{\nu+ d-1}\rp^{2}\\\nonumber
    &\quad\quad\quad+\frac{ d}{2\sigma^{2}}\lp\hat{\mu}_{\nu:\nu+ d-1}-\frac{\lp\nu-1\rp\hat{\mu}_{1:\nu-1}+ d\hat{\mu}_{\nu:\nu+ d-1}}{\nu+ d-1}\rp^{2}<\tilde{\beta}_{\mrm{GSR}}\lp\nu+ d-1,\delta_{\mrm{F}}\rp\Bigg)\\
    &=\Pr_{\nu}\lp\frac{\lp\nu-1\rp d}{2\sigma^{2}\lp\nu+ d-1\rp}\lp\hat{\mu}_{1:\nu-1}-\hat{\mu}_{\nu:\nu+ d-1}\rp^{2}<\tilde{\beta}_{\mrm{GSR}}\lp\nu+ d-1,\delta_{\mrm{F}}\rp\rp\label{eq:gsr-ld-1}
\end{align}
where step $(a)$ stems from Lemma \ref{lem:GLR-kl-unknown-pre-post} and step $(b)$ results from $\lbp\nu+ d-1\rbp\subseteq\lbp1,\dots,\nu+ d-1\rbp$.  Recall that $\mu_{0}$ and $\mu_{1}$ are the pre- and post-change means, and that the definition of $d$ is given in \eqref{eq:d}. In order to apply Lemma \ref{lem:sub_Gaussian_diff}, we need to convert $\lp\hat{\mu}_{1:\nu-1}-\hat{\mu}_{\nu:\nu+d-1}\rp^{2}$ in the last line of \eqref{eq:gsr-ld-1} into $\lp\lp\hat{\mu}_{1:\nu-1}-\hat{\mu}_{\nu:\nu+d-1}\rp-\lp\mu_{0}-\mu_{1}\rp\rp^{2}$. Following the same steps in \eqref{eq:implication_delay} with $\beta = \tilde{\beta}_{\mrm{GSR}}$, we can show that the event $\lbp\frac{\lp\nu-1\rp d}{2\sigma^{2}\lp\nu+d-1\rp}\lp\hat{\mu}_{1:\nu-1}-\hat{\mu}_{\nu:\nu+d-1}\rp^{2}<\tilde{\beta}_{\mrm{GSR}}\lp \nu+d-1,\delta_{\mrm{F}}\rp\rbp$ implies the event $\lbp\frac{\lp\nu-1\rp d}{2\sigma^{2}\lp\nu+d-1\rp}\lp\lp\hat{\mu}_{1:\nu-1}-\hat{\mu}_{\nu:\nu+d-1}\rp-\lp\mu_{0}-\mu_{1}\rp\rp^{2}\geq\tilde{\beta}_{\mrm{GSR}}\lp \nu+d-1,\delta_{\mrm{F}}\rp\rbp$ for any $\nu\in\lbp m+1,\dots,T-d\rbp$ with the choices of $m$ and $d$ in \eqref{eq:m} and \eqref{eq:d}, respectively. 
Then, for any $\nu\in\lbp m+1,\dots,T-d\rbp$,
\begin{align}\nonumber
    &\Pr_{\nu}\lp\tilde{\tau}_{\mrm{GSR}}\geq\nu+d\rp\\\nonumber
    &\leq\Pr_{\nu}\lp\frac{\lp\nu-1\rp d}{2\sigma^{2}\lp\nu+d-1\rp}\lp\lp\hat{\mu}_{1:\nu-1}-\hat{\mu}_{\nu:\nu+d-1}\rp-\lp\mu_{0}-\mu_{1}\rp\rp^{2}\geq\tilde{\beta}_{\mrm{GSR}}\lp\nu+d-1,\delta_{\mrm{F}}\rp\rp\\\nonumber
    &\overset{(a)}{\leq}\Pr_{\nu}\lp\frac{\lp\nu-1\rp d}{2\sigma^{2}\lp\nu+d-1\rp}\lp\lp\hat{\mu}_{1:\nu-1}-\hat{\mu}_{\nu:\nu+d-1}\rp-\lp\mu_{0}-\mu_{1}\rp\rp^{2}\geq\tilde{\beta}_{\mrm{GSR}}\lp m+d,\delta_{\mrm{F}}\rp\rp\\\nonumber
    &\overset{(b)}{\leq}\Pr_{\nu}\lp\frac{\lp\nu-1\rp d}{2\sigma^{2}\lp\nu+d-1\rp}\lp\lp\hat{\mu}_{1:\nu-1}-\hat{\mu}_{\nu:\nu+d-1}\rp-\lp\mu_{0}-\mu_{1}\rp\rp^{2}\geq\frac{5}{2}\log\lp\frac{4\lp m+d\rp^{3/2}}{\delta_{\mrm{F}}}\rp\rp\\\nonumber
    &=\Pr_{\nu}\lp\frac{\lp\nu-1\rp d}{\nu+d-1}\lp\lp\hat{\mu}_{1:\nu-1}-\hat{\mu}_{\nu:\nu+d-1}\rp-\lp\mu_{0}-\mu_{1}\rp\rp^{2}\geq2\sigma^{2}\log\lp\frac{2^{5}\lp m+d \rp^{15/4}}{\delta_{\mrm{F}}^{5/2}}\rp\rp\\\nonumber
    &\overset{(c)}{\leq}\frac{\delta_{\mrm{F}}^{5/2}}{16\lp m+d\rp^{15/4}}\\\nonumber
    &\leq\frac{\delta_{\mrm{F}}^{5/2}}{16\lp\delta_{\mrm{F}}^{5/2}2^{-16/15}\delta^{-4/15}_{\mrm{D}}\rp^{15/4}}\\
    &=\delta_{\mrm{D}}
\end{align}
where step $(a)$ is due to the fact that $\tilde{\beta}_{\mrm{GSR}}\lp n,\delta_{\mrm{F}}\rp$ is increasing with $n$, whereas step $(b)$ is owing to the fact that $\tilde{\beta}_{\mrm{GSR}}\lp n,\delta_{\mrm{F}}\rp\geq\frac{5}{2}\log\lp4n^{3/2}/\delta_{\mrm{F}}\rp$. Step $(c)$ comes from Lemma \ref{lem:sub_Gaussian_diff}. 

\end{document}